\documentclass[aps,amsmath,amssymb,superscriptaddress,a4paper,onecolumn,accepted=2020-06-30]{quantumarticle}
\pdfoutput=1

\usepackage[utf8]{inputenc}
\usepackage[english]{babel}
\usepackage[T1]{fontenc}
\usepackage{amsmath}
\makeatletter
\newcommand{\pushright}[1]{\ifmeasuring@#1\else\omit\hfill$\displaystyle#1$\fi\ignorespaces}
\newcommand{\pushleft}[1]{\ifmeasuring@#1\else\omit$\displaystyle#1$\hfill\fi\ignorespaces}
\makeatother
\usepackage{hyperref}

\usepackage{tikz}
\usepackage[numbers,sort&compress,square]{natbib}
\usepackage{amssymb}
\usepackage{color}
\usepackage{amsbsy}
\usepackage{amsthm}
\usepackage{graphicx}
\usepackage{bbm}
\usepackage{bm}
\usepackage{epsfig}
\usepackage{xfrac}
\usepackage{xcolor}
\newtheorem{theorem}{Theorem}

\newcommand{\ket}[1]{\left| {#1} \right\rangle}
\newcommand{\bra}[1]{\left\langle {#1}\right|}

\newcommand{\braket}[2]{\langle #1|#2\rangle}

\newcommand{\abs}[1]{\left| {#1} \right|}

\renewcommand{\t}[1]{\textrm{#1}}
\newcommand{\trace}[0]{\mathrm{Tr}}

\newcommand{\real}{\mathrm{Re}}
\newcommand{\imag}{\mathrm{Im}}

\newcommand{\vspan}{\mathrm{span}}
\newcommand{\bomega}{{\boldsymbol{\omega}}}
\newcommand{\bsigma}{{\boldsymbol{\sigma}}}
\newcommand{\fisher}{F}
\newcommand{\cov}{\Sigma}
\newcommand{\mH}{\mathcal{H}}
\newcommand{\mS}{\mathcal{S}}
\newcommand{\mC}{\mathcal{C}}
\newcommand{\mHC}{\mathcal{H}_{\mathcal{C}}}
\newcommand{\mHR}{{\mathcal{H}_{\mathcal R}}}
\newcommand{\mHCi}{\mathcal{H}_{\mathcal{C}_i}}
\newcommand{\id}{\openone}

\newcommand{\psiw}{\psi_\bomega}
\newcommand{\psiwn}{\psi_{\bomega=\boldsymbol{0}}}
\newcommand{\weight}{W}
\newcommand{\cost}{{\Delta_\weight^2\tilde\bomega}}

\newcommand{\up}{\uparrow}
\newcommand{\down}{\downarrow}
\newcommand{\uup}{\ket{\Phi_+}}
\newcommand{\uum}{\ket{\Phi_-}}
\newcommand{\udp}{\ket{\Psi_+}}
\newcommand{\udm}{\ket{\Psi_-}}
\newcommand{\bR}{{\mathbb R}}

\newcommand{\A}{\Gamma}
\newcommand{\D}{D}

\newcommand{\SEP}{{\rm SEP}}
\newcommand{\JNT}{{\rm JNT}}

\usepackage{braket}
\newcommand{\gv}[1]{\ensuremath{\text{\boldmath$ #1 $}}}
\newcommand{\eff}{\mathcal R} %
\newcommand{\ceff}{\mathcal C} %

\newcommand{\thmref}[1]{\hyperref[#1]{Theorem~\ref{#1}}}
\newcommand{\lemmaref}[1]{\hyperref[#1]{Lemma~\ref{#1}}}
\newcommand{\figref}[1]{\hyperref[#1]{Fig.~\ref{#1}}}
\newcommand{\figaref}[1]{\hyperref[#1]{Fig.~\ref{#1}a}}
\newcommand{\figbref}[1]{\hyperref[#1]{Fig.~\ref{#1}b}}
\newcommand{\figcref}[1]{\hyperref[#1]{Fig.~\ref{#1}c}}
\renewcommand{\eqref}[1]{\hyperref[#1]{Eq.~(\ref{#1})}}
\newcommand{\eqsref}[2]{\hyperref[#1]{Eqs.~(\ref{#1})-(\ref{#2})}}
\newcommand{\appref}[1]{\hyperref[#1]{Appx.~\ref{#1}}}
\newcommand{\secref}[1]{\hyperref[#1]{Sec.~\ref{#1}}}

\usepackage{dsfont}

\title{Optimal probes and error-correction schemes in multi-parameter quantum metrology}
\author{Wojciech G{\'{o}}recki*}
\affiliation{Faculty of Physics, University of Warsaw, Pasteura 5, 02-093 Warsaw, Poland}
\orcid{0000-0001-9912-9186}
\author{Sisi Zhou*}
\affiliation{Departments of Applied Physics and Physics, Yale University, New Haven, Connecticut 06511, USA}
\affiliation{Yale Quantum Institute, Yale University, New Haven, Connecticut 06511, USA}
\affiliation{Pritzker School of Molecular Engineering, University of Chicago, Chicago, IL 60637, USA}
\orcid{0000-0003-4618-8590}
\author{Liang Jiang}
\affiliation{Departments of Applied Physics and Physics, Yale University, New Haven, Connecticut 06511, USA}
\affiliation{Yale Quantum Institute, Yale University, New Haven, Connecticut 06511, USA}
\affiliation{Pritzker School of Molecular Engineering, University of Chicago, Chicago, IL 60637, USA}
\orcid{0000-0002-0000-9342}
\author{Rafa{\l} Demkowicz-Dobrza{\'n}ski}
\affiliation{Faculty of Physics, University of Warsaw, Pasteura 5, 02-093 Warsaw, Poland}
\orcid{0000-0001-5550-4431}

\begin{document}

\maketitle

\let\thefootnote\relax\footnotetext{*These two authors provided key and equal contributions to the project.}

\begin{abstract}
We derive a necessary and sufficient condition for the possibility of achieving the Heisenberg scaling in general adaptive multi-parameter estimation schemes in presence of Markovian noise.
In situations where the Heisenberg scaling is achievable, we provide a semidefinite program to identify the optimal quantum error correcting (QEC) protocol that yields the best estimation precision.
We overcome the technical challenges associated with potential incompatibility of the measurement optimally extracting information on different parameters by utilizing the Holevo Cram{\'e}r-Rao (HCR) bound for pure states. We provide examples of significant advantages offered by our joint-QEC protocols, that sense all the parameters utilizing a single error-corrected subspace, over separate-QEC protocols where each parameter is effectively sensed in a separate subspace.
\end{abstract}

\section{Introduction}
\begin{figure}[t]
\center
\includegraphics[width=0.6\columnwidth]{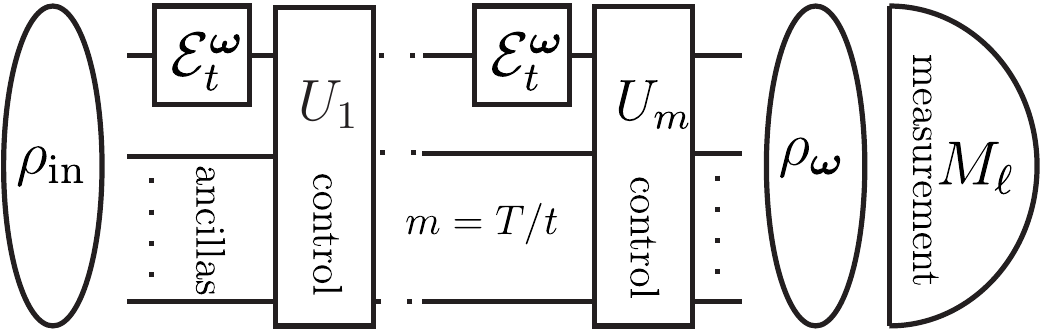}
\caption{General adaptive mutli-parameter quantum metrological scheme, where $P$ parameters
$\boldsymbol{\omega} = [\omega_i]_{i=1}^P$ are to be estimated. Total probe system evolution time $T$ is divided into a number $m$ of $t$-long steps of probe evolution $\mathcal{E}_t^{\boldsymbol{\omega}}$ interleaved with general unitary controls $U_i$. In the end a general POVM $\{M_{\ell}\}$ is performed yielding estimated value of all parameters $\tilde{\boldsymbol{\omega}}(\ell)$ with probability $p(\ell) = \t{Tr}(\rho_{\boldsymbol{\omega}} M_{\ell})$.
}
\label{fig:scheme}
\end{figure}
Quantum metrology aims at exploiting all possible features of quantum systems, such as coherence or entanglement, in order to boost the precision of measurements beyond that achievable by metrological schemes that operate within classical or semi-classical paradigms~\cite{giovannetti2006quantum,Paris2009, giovannetti2011advances,Toth2014, Demkowicz2015, Schnabel2016, degen2017quantum,Pezze2018, Pirandola2018}.
The most persuasive promise of quantum metrology is the possibility of obtaining the so-called Heisenberg scaling (HS), which manifests itself in the quadratically improved scaling of precision as a function of number of elementary probe systems involved in the experiment~\cite{caves1981quantum,holland1993interferometric,lee2002quantum,wineland1992spin,mckenzie2002experimental,bollinger1996optimal,
leibfried2004toward,giovannetti2004quantum,huelga1997improvement,Berry2000} or the total interrogation time of a probe system~\cite{de2005quantum}.
In either of these cases, the presence of decoherence typically restricts the quadratic improvement to a small particle number or a short-time regime, whereas in the asymptotic regime the quantum-enhancement amounts to constant factor improvements \cite{fujiwara2008fibre,demkowicz2009quantum,escher2011general,demkowicz2012elusive,kolodynski2013efficient,knysh2014true} even under the most general adaptive schemes \cite{demkowicz2014using}. Still, there are specific models where even in the presence of decoherence the asymptotic HS is achievable via application of appropriate quantum error correction (QEC) protocols~\cite{kessler2014quantum,dur2014improved,ozeri2013heisenberg,arrad2014increasing,unden2016quantum, reiter2017dissipative,sekatski2017quantum,demkowicz2017adaptive,zhou2018achieving,layden2018spatial,layden2018ancilla,kapourniotis2019fault,tan2019quantum,zhou2020theory,chen2020fluctuation}.

Recently, a general theory providing a necessary and sufficient condition, the HNLS condition (an acronym for ``Hamiltonian-Not-in-Lindblad-Span''), for achieving the HS in a finite-dimensional system in the most general adaptive quantum metrological protocols under Markovian noise, has been developed~\cite{demkowicz2017adaptive,zhou2018achieving}. The theory allows for a quick identification of the most promising quantum metrological models and provides a clear recipe for designing the optimal adaptive schemes based on appropriately tailored QEC protocols. However, HNLS is restricted to the single-parameter estimation case, while a lot of relevant metrological problems, like vector field sensing  (e.g. magnetic field) \cite{baumgratz2016quantum}, imaging \cite{tsang2016quantum}, multiple-arm interferometry \cite{humphreys2013quantum,Gessner2018} or waveform estimation \cite{Tsang2011, Berry2013} are inherently multi-parameter estimation problems. Multi-parameter estimation problems drew a lot of attention in recent years~\cite{matsumoto2002new,genoni2013optimal,ragy2016compatibility,yuan2016sequential,Kura2018, Liu2017, Nichols2018,ge2018distributed}, yet  the fundamental questions regarding the  achievability of the HS as well as the form of the optimal metrological protocols in multiple-parameter estimation in presence of noise have not been answered so far. The aim of this paper is to fill this gap.

The methods developed for the single parameter estimation case, in particular the semidefinite program (SDP)
that allows to identify the  optimal QEC protocol~\cite{zhou2018achieving}, are not applicable in the multi-parameter regime. The reasons are threefold.

First, the widely used quantum Cram{\'e}r-Rao (CR) bound for multiple parameters is not in general saturable, due to the incompatibility of the optimal measurements for different parameters~\cite{ragy2016compatibility, yuan2016sequential, matsumoto2002new}.
Therefore, unlike in the single-parameter case, the quantum Fisher information (QFI) does not provide the full insight into the problem~\cite{ragy2016compatibility,braunstein1994statistical,helstrom1976quantum,Holevo1982,demkowicz2020multi}.
On the other hand, stronger bounds, such as the HCR bound~\cite{Holevo1982,nagaoka2005asymptotic,suzuki2016explicit,Guta2007,yamagata2013quantum,demkowicz2020multi}, have no closed-form expressions (except for specific cases \cite{fujiwara1994multi}). Moreover, the HCR bound, although solvable via an SDP~\cite{albarelli2019evaluating}, does not shed light on the corresponding optimal measurements saturating it. In general, the HCR bound is only saturable when collective measurements on all copies of quantum states are allowed, which is a demanding condition in practice~\cite{Guta2007,yamagata2013quantum}. As a result, the optimal measurements on the output quantum states are hard to identify.

Second, there is no general recipe to find the optimal input states in multi-parameter estimation even in the noiseless case~\cite{yuan2016sequential,Kura2018}, unlike 
in the single-parameter estimation case where the optimal input state is simply the equally weighted superposition between the eigenstates corresponding to the maximum and minimum eigenvalues of the Hamiltonian.

Finally, in the single-parameter case~\cite{zhou2018achieving}, all valid two-dimensional QEC codes were mapped into a traceless Hermitian matrix representing the difference between logical zero and one codes. In the multi-parameter case, however, it is not clear whether valid QEC codes could be mapped into a convex set when the code dimension is large, which is inevitable in multi-parameter estimation.

In this paper, we generalize the HNLS condition to multi-parameter scenarios, and provide an SDP to find the best possible QEC protocol
in the situations when the HNLS condition is satisfied (including all noiseless cases).
The solution yields an explicit form of the optimal input state, QEC codes and measurements. No collective measurements are required on the output states. Our protocol goes beyond the typically used QFI-based formalism and overcomes all the challenges related with the multi-parameter aspect of the problem mentioned above. Our work reveals the advantage of QEC protocols in multi-parameter estimation and we expect that the SDP formulation of our problem will also be an inspiration for other research areas in quantum error correction and quantum metrology.

\section{Formulation of the model}

We assume the dynamics of a $d$-dimensional probe system $\mH_S$ is given by a general quantum master equation~\cite{lindblad1976generators,gorini1976completely,breuer2002theory}:
\begin{equation}
\label{eq:evol}
\frac{d\rho}{dt} = -i[H,\rho]+\sum_{k=1}^r(L_k \rho L_k^\dagger - \frac{1}{2}\{L_k^\dagger L_k,\rho\}),
\end{equation}
where the parameters to be estimated $\boldsymbol{\omega}=[\omega_1,\dots,\omega_P]$ enter linearly into the Hamiltonian of the evolution via Hermitian generators $\boldsymbol{G} = [G_1,\dots,G_P]^T$ (where $^T$ denotes transpose) so that $H = \boldsymbol{\omega}\cdot\boldsymbol{G}\equiv\sum_{k=1}^P \omega_k G_k$, and $L_k$ are operators representing a general Markovian noise.
Similar to the previous investigations~\cite{demkowicz2017adaptive,zhou2018achieving} we consider the most general adaptive scheme (see \figref{fig:scheme})~\cite{demkowicz2014using} with an unlimited number of ancillae (denoted jointly as $\mH_A$), instantaneous perfect intermediate unitary operations $U_i$ and a general POVM on the final state $\rho_{\boldsymbol{\omega}}$. $\mathcal{E}_t^{\boldsymbol{\omega}}$ represents the probe system dynamics integrated over time $t$, whereas the total probe interrogation time is $T$. Such schemes are the most general schemes of probing quantum dynamics, assuming the total interrogation time is $T$, and encompass in particular all QEC procedures.

In single-parameter estimation the optimal protocol is the one that yields the minimum estimation variance. In multi-parameter case the estimator covariance matrix is the key object capturing estimation precision, defined as~\cite{helstrom1976quantum,Holevo1982}:
\begin{equation}
\label{eq:cov}
\cov_{ij} = \textstyle{\sum}_{\ell}\,  \trace(\rho_{\boldsymbol{\omega}} M_\ell) (\tilde{\omega}_i(\ell) - \omega_i)(\tilde{\omega}_j(\ell) - \omega_j)
\end{equation}
for $i,j=1,\ldots,P$,
where $M_\ell \geq 0$, $\sum_\ell M_{\ell} = \openone$, are measurement operators (``$\geq 0$'' for matrices means positive semidefinite) and
 $\tilde{\gv{\omega}}(\ell)$ is an estimator function mapping a measurement result $\ell$ to the parameter space. 

Diagonal entries of $\cov$ represent variances of estimators of respective parameters while off-diagonal terms represent covariance between different parameters. As a figure of merit one may simply choose $\trace(\cov)$ which will be the sum of all individual parameter variance, or  more generally $\trace(\weight \cov )$, where  $\weight$ is a real positive cost matrix that determines the weight we associate with each parameter in the effective scalar cost function
\begin{equation}
\cost\equiv\trace (\weight \cov ).
\end{equation}
Note that we require strict positivity of $\weight$ which is equivalent to saying that it is an estimation problem of all $P$ parameters, and not a problem where effectively only a smaller number of parameters are relevant. We assume the measurement-estimation strategy to be locally unbiased at some fixed parameter point $\bomega$, i.e.
 \begin{equation}
 \label{eq:unbiased}
 \sum_\ell \tilde{\omega}_j (\ell)\trace(\rho_{\gv \omega}M_\ell) = \omega_j, \quad  \sum_\ell \tilde{\omega}_j(\ell)\trace\big(
\partial_i \rho_{\gv \omega} M_\ell\big) = \delta_{ij},
\end{equation}
where $\partial_i \rho_{\gv \omega} = \frac{\partial \rho_{\gv \omega}}{\partial \omega_i}$,
which is a standard assumption necessary to obtain meaningful precision bounds within the frequentist estimation framework~\cite{Kay1993,gill2000state}.

\section{The necessary and sufficient condition for the HS}
\label{sec:HNLS}

We say that the HS in a multi-parameter estimation problem is achieved when there exists an adaptive protocol such that for every $\weight>0$, $\cost\propto 1/T^2$ in the limit $T \rightarrow \infty$. This is equivalent to a requirement that all parameters (and any combination of parameters) are estimated with precision that scales like the HS. The following theorem generalize the HNLS condition \cite{demkowicz2017adaptive,zhou2018achieving} to multi-parameter scenarios. 


\begin{theorem}[Multi-parameter HNLS]
The HS can be achieved in a multi-parameter estimation problem if and only if $\{(G_i)_\perp, i=1,\dots, P\}$ are linearly independent operators. Here $(G_i)_\perp$ are orthogonal projections of $G_i$ onto space $\mS^\perp$ which is the orthogonal complement of the Lindblad span
\begin{equation}
 \mS = \mathrm{span}_{\mathbb{R}}\{\id,L_k^{{\mathrm H}}, i L_k^{{\rm AH}},(L_k^\dagger L_j)^{{\mathrm H}}, i(L_k^\dagger L_j)^{{\rm AH}},\,\forall j,k\},
\end{equation}
in the Hilbert space of Hermitian matrices under the standard Hilbert-Schmidt scalar product,  whereas 
the superscripts $^{{\mathrm H}}$, $^{{\rm AH}}$ denote the Hermitian and anti-Hermitian part of an operator respectively.
\end{theorem}
\begin{proof}
Let us start with a brief review of the HNLS condition in the single-parameter case, where $H = \omega G$ involves only a single generator $G$.
As shown in~\cite{demkowicz2017adaptive,zhou2018achieving}, the necessary and sufficient condition to achieve the HS is that $G \notin \mS$, or in other words that $G_\perp \neq 0$.
In particular, following \cite{zhou2018achieving} (see the section named ``QEC code for HL scaling when HNLS holds''), an explicit construction of the optimal QEC code was provided, where the code space $\mH_\mC\subseteq\mH_S\otimes\mH_A$ is defined on the Hilbert space of the probe system $\mH_S$ extended by an ancillary space $\mH_A \cong \mH_S$.
The code space satisfies the QEC condition~\cite{zhou2018achieving,knill1997theory}:
\begin{equation}
\Pi_{\mHC} (S\otimes\openone) \Pi_{\mHC} \propto \Pi_{\mHC}, \forall S\in\mS,
\label{eq:QECs}
\end{equation}
where the operator $S$ acting on $\mH_S$ was tensored with identity on $\mH_A$ and $\Pi_{\mHC}$ denotes the projection onto $\mHC$.
Metrological sensitivity is guaranteed by the fact that $G$ acts non-trivially on $\mHC$:
\begin{equation}
G^{\ceff} = \Pi_{\mHC} (G\otimes\openone)\Pi_{\mHC} \not\propto \Pi_{\mHC}.
\end{equation}
As a result we obtain a noiseless unitary evolution generated by $G^{\ceff}$ leading to
the HS in the estimation precision of $\omega$.

(Necessity) Suppose $(G_i)_\perp$'s are linearly dependent. Then there exists a linear (invertible) transformation on the parameter space $A \in \mathbb{R}^{P\times P}$: $\bomega' = \bomega A^{-1}$, (where we also modify accordingly the generators $\boldsymbol{G}' = A \boldsymbol{G}$ and the cost matrix $\weight'=A\weight A^{T}$, so that $H$ and $\cost$ remain unchanged), such that $(G'_i)_\perp = 0$ for some $i$. Then, from the single-parameter theorem, $\omega'_i$ cannot be estimated with precision better than $\Delta^2\tilde{\omega}'_i\sim 1/T$ which contradicts the HS requirements.

(Sufficiency) Suppose $(G_i)_\perp$'s are linearly independent.
We assume the ancillary space to be a direct sum of $P$ subspaces $\mH_{A_i}$ so that the whole Hilbert space is $\mH_S \otimes (\mathcal{H}_{A_1} \oplus \cdots \oplus \mathcal{H}_{A_P})$ (see \figaref{fig:diagram}). We may construct separate code spaces for each parameter using orthogonal ancillary subspace $\mHCi\subseteq\mH_S \otimes\mathcal{H}_{A_i}$ so that the QEC conditions \eqref{eq:QECs} are satisfied within each code space $\mHCi$ separately.  While constructing the code space for the $i$-th parameter, we include all the  remaining generators $G_j$ ($j \neq i$) in the Lindblad span, so effectively treating them as noise
i.e.  $\mS_i = \mathrm{span}_{\mathbb{R}}\{\{\id,L_k^{{\mathrm H}}, i L_k^{{\rm AH}},(L_k^\dagger L_j)^{{\mathrm H}}, i(L_k^\dagger L_j)^{{\rm AH}}\}_{j,k}\cup \{G_j\}_{j\neq i}\}$. As a result thanks to the QEC condition it follows that $\forall_{i\neq j}\Pi_{\mHCi}\left(G_j\otimes\openone\right)\Pi_{\mHCi} \propto  \Pi_{\mHCi}$ and hence within a given subspace only one parameter is being sensed via the effective generator $G_i^{\ceff_i} = \Pi_{\mHCi} (G_i\otimes\openone) \Pi_{\mHCi}$, while all other generators act trivially.
If $\ket{\psi_i}\in\mHCi$ is the optimal state for measuring $\omega_i$, the state to be used in order to obtain HS for all parameters which is not affected by noise reads $\rho_\text{in} = \frac{1}{P}\sum_{i=1}^{P}\ket{\psi_i}\bra{\psi_i} \in \mH_S \otimes \left(\bigoplus_{i=1}^P \mathcal{H}_{A_i}\right)$---there is no measurement incompatibility issue because different parameters are encoded on orthogonal subspaces.
\end{proof}

Similar to the single-parameter case, it must be admitted that in realistic situations with generic noise, HNLS is often  violated~\cite{demkowicz2012elusive, demkowicz2017adaptive}.
Therefore, a more pragmatic approach is required taking into account the fact that in a real experiment the total time of evolution $T$ is always finite. Let us consider a situation where $H \in \mathcal{S}$, but where some noise components are weak~\cite{zhou2018achieving}. Specifically, we divide Lindblad operators in \eqref{eq:evol} into two sets---strong noise generators $\{L_k\}$ and weak noise generators $\{J_m\}$ satisfying $\epsilon:=\|\sum_m J^\dagger_m J_m \|$ where $\|\cdot\|$ denotes operator norm. If the HNLS condition is satisfied for the strong noise part, we could choose the code space $\mHC$ which allows to completely erase the strong noise $\{L_k\}$ and the resulting effective noise rate would be upper bounded by $\epsilon$~\cite{zhou2018achieving}. This means that the distance between state of the error-corrected probe and the state evolving under ideal noiseless evolution will be of the order $\Theta(\epsilon T)$. Therefore, for sufficiently short evolution, $T=o(1/\epsilon)$, the precision of estimation will still scale quadratically with the total time $\cost\propto \frac{1}{T^2}$ whereas for larger $T$, it will gradually approach the standard $1/T$ scaling.

\begin{figure}[t]
\center
\includegraphics[width=0.9\columnwidth]{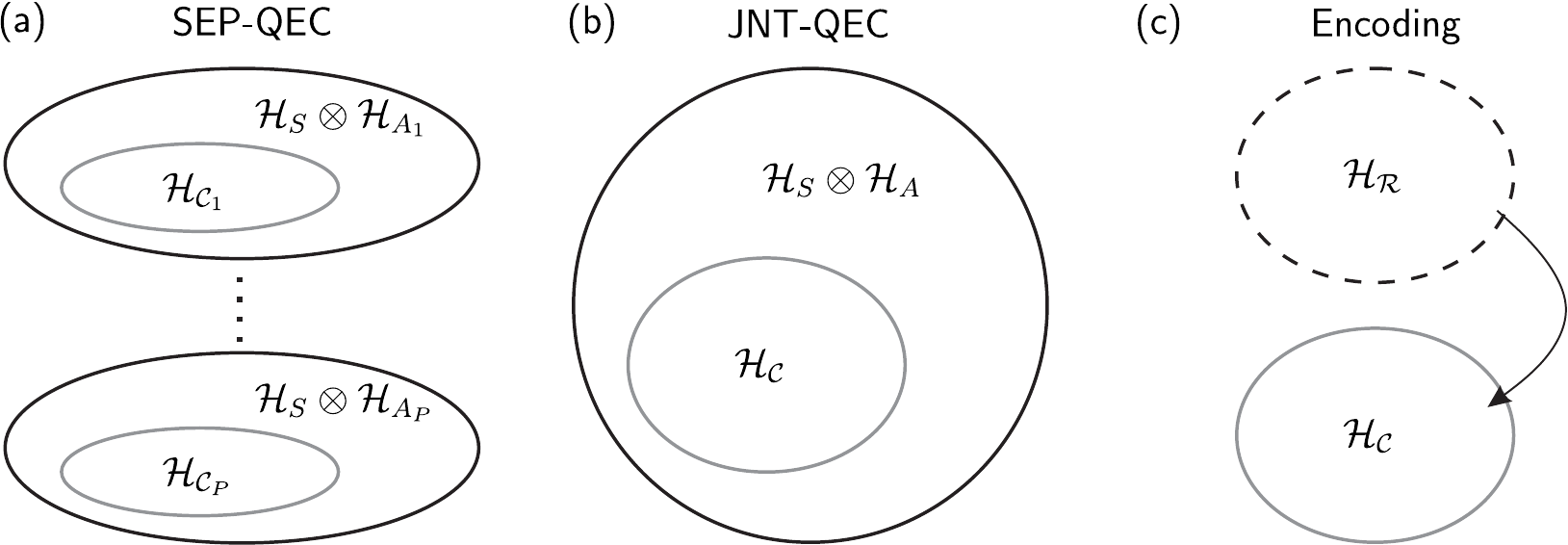}
\caption{Schematic diagrams of relations between Hilbert spaces $\mH_S,\mH_A,\mH_\mC,\mH_{A_i},\mH_{\mC_i},\mH_\eff$. (a) In SEP-QEC, we use $P$ mutually orthogonal ancillary subspaces $\mH_{A_i}$ to sense each parameter $\omega_i$. $\mH_A = \bigoplus_{i=1}^P \mH_{A_i}$ and $\mH_\mC = \bigoplus_{i=1}^P \mH_{\mC_i}$. $\dim(\mH_{A_i}) = \dim(\mH_S) = d$ and $\dim(\mH_{\mC_i}) = 2$. (b) In JNT-QEC, we use a single code space $\mH_\mC \subseteq \mH_S \otimes \mH_A$ to estimate all parameters jointly. $\dim(\mH_{A}) = (P+1)d$ and $\dim(\mH_\mC) = P + 1$. (c) We use $\mH_\eff$ to represent the logical space ${\rm span}\{\ket{0},\ket{1},\ldots,\ket{P}\}$ which is encoded into the physical space $\mH_\mC = {\rm span}\{\ket{c_0},\ket{c_1},\ldots,\ket{c_P}\}$.
}
\label{fig:diagram}
\end{figure}

\section{Optimal probes, error correction schemes and measurements}
\label{sec:optprobes}

In \secref{sec:HNLS}, we provided a QEC code where each parameter is sensed separately in different error-corrected subspaces (see \figaref{fig:diagram}). Such protocols will be referred as separate-QEC protocols (SEP-QEC). In contrast to this construction, we will now consider QEC strategies which allow for simultaneous estimation of all the parameters in a single coherent protocol by utilizing states within a single protected code space, which we will call the joint-parameter QEC scheme (JNT-QEC).
In this section we provide a general method to identify the optimal JNT-QEC, while its potential advantages
over the optimal SEP-QEC will be discussed in \secref{sec:JNTSEP}, as well as in \secref{sec:examples}
when studying concrete estimation models.

 From now on, we assume the multi-parameter HNLS condition is satisfied. Without loss of generality, we assume the generators $\{G_i\}_{i=1}^P \subset \mathcal S^\perp$ are orthonormal, since the components in $\mathcal S$ do not contribute and there is always a linear transformation $A$ on parameters leading to orthonormality.
The following theorem provides a recipe to find the optimal JNT-QEC.

\begin{theorem}[Optimal JNT-QEC]
Given a cost matrix $W$. 
If the multi-parameter HNLS condition is satisfied with generators $\{G_i\}^{P}_{i=1} \subset S^{\perp}$,
the minimum cost $\cost$ that can be achieved in a JNT-QEC reads
\begin{equation}
\begin{split}
\label{eq:thm2}
&\cost=
\frac{P}{4T^2} \min_{G^\eff_i,B_i,\nu_i,K,w} w,\\
\text{~subject to~~~~~}
&\id_{\scriptscriptstyle P+1}\otimes \frac{\id_d}{d}  +  \sum_{i=1}^{P}  (G^\eff_i)^T  \otimes  G_{i}  +  \sum_{i=P+1}^{P'}  \nu_i \id_{\scriptscriptstyle P+1}  \otimes  S_i  +  \sum_{i=P^\prime +1}^{d^2-1}  B_i  \otimes  R_i  \geq  0,\\
&
\A_{ij}=\imag[G^\eff_j]_{i0},\quad
\begin{pmatrix}
    w\openone_{\scriptscriptstyle P}       &  K\\
    K       &  \openone_{\scriptscriptstyle P}
\end{pmatrix}\geq 0,\quad
\begin{pmatrix}
    K       &  \openone_{\scriptscriptstyle P}\\
    \openone_{\scriptscriptstyle P}       &  \A\sqrt{\weight^{-1}}
\end{pmatrix}\geq 0,
\end{split}
\end{equation}
where $\id_d/\sqrt{d}$, $\{G_{i}\}_{i=1}^P$, $\{S_{i}\}_{i=P+1}^{P'}$, $\{R_{i}\}_{i=P'+1}^{d^2-1}$ form an orthonormal basis of Hermitian operators  acting on $\mH_S$ such that $\mS = \mathrm{span}_{\mathbb{R}}\{\id_d, {(S_{i})}_{i=P+1}^{P'}\}$.
Moreover, $G^\eff_i$, $B_i$ are Hermitian $(P+1)\times (P+1)$ matrices (with matrix indices taking values from $0$ to $P$),
and $\nu_i\in \mathbb R$. $\A$ and $K$ are real $P\times P$ matrices (with matrix indices from $1$ to $P$).
\end{theorem}
As a semidefinite program (SDP) it could be easily solved numerically, for example, using the Matlab-based package CVX~\cite{grant2008cvx}. Before giving a formal proof, let us briefly review some existing bounds in multi-parameter metrology and discuss their saturability.

\subsection{General bounds in multi-parameter metrology}
\label{sec:bounds}

In this section, we discuss bounds for general parameter estimation problems on fixed quantum states. We use $\mH$ to denote a general finite-dimensional Hilbert space.

Most commonly, quantum multi-parameter estimation problems are analyzed utilizing the standard quantum CR bound~\cite{helstrom1976quantum,braunstein1994statistical,Holevo1982}:
\begin{equation}
\cost \geq \trace (\weight \fisher^{-1}), \quad \fisher_{ij}=\real(\trace(\rho_\bomega\Lambda_i\Lambda_j)),
\end{equation}
where $\fisher$ is a $P\times P$ QFI matrix and $\Lambda_i$ (symmetric logarithmic derivatives) satisfy $\partial_i\rho_\bomega=\frac{1}{2}(\Lambda_i\rho_\bomega+\rho_\bomega\Lambda_i)$.
This bound is not saturable in general, due to potential non-compatibility of the optimal measurements, unless $\imag [\trace(\rho_\bomega\Lambda_i\Lambda_j)]=0$~\cite{ragy2016compatibility}.
Moreover, a direct minimization of the CR bound over all JNT-QEC,  with the saturability constraint imposed, does not necessarily guarantee the identification of the optimal protocol---there is a possibility that the optimal protocol does not meet the saturability constraint for the CR bound.


Therefore, in order to identify truly optimal protocols we need to resort to a stronger HCR bound~\cite{Holevo1982,nagaoka2005asymptotic,suzuki2016explicit}:
\begin{equation}
\begin{split}
\label{eq:holevo}
&\cost \geq \min_{\{X_i\}}(\trace\left(\weight \cdot {\rm Re} V)+\trace\left[{\rm abs}(\weight \cdot \imag V)\right] \right),\quad\text{where}~ V_{ij}=\trace(X_i X_j \rho_{\bomega}),\\
&\text{for Hermitian}~X_i \in{\cal L}(\mH),\quad \text{subject~to} ~\trace(X_i\partial_j\rho_{\bomega}) = \delta_{ij},
\end{split}
\end{equation}
where ${\cal L}(\circ)$ denotes the set of all linear operators acting on $\circ$, ${\rm Re}$ and ${\rm Im}$ denote the real and imaginary part of a matrix (not to be confused with the Hermitian and anti-Hermitian part of a matrix), and $\trace[{\rm abs}(\cdot)]$ is the sum of absolute values of the eigenvalues of a matrix. When the second term is dropped the HCR bound reduces to the standard CR bound~\cite{Holevo1982,ragy2016compatibility}. Unlike the CR bound this bound is
saturable in general using collective measurements on many copies \cite{Guta2007,yamagata2013quantum}. 
On the other hand, the HCR is defined via an optimization problem, making it usually more difficult to deal with than the standard quantum CR bound with a closed-form expression.

In the case of pure states $\rho_{\bomega} = \ket{\psiw}\bra{\psiw}$, however, we note that the HCR is exactly equivalent to~\cite{matsumoto2002new}:
\begin{equation}
\label{eq:matsumoto}
\begin{split}
&\cost\geq\min_{\{\ket{x_i}\}}\trace(\weight V),\quad {\textrm{where}}~~V_{ij}=\braket{x_i|x_j},\\
&\text{for}~\ket{x_i} \in \vspan\{\ket{\psiw},\partial_1\ket{\psiw},...,\partial_P\ket{\psiw}\}\oplus \mathbb C^P,\\
&\text{subject~to}~2\real (\braket{x_i|\partial_j|\psi_\bomega})=\delta_{ij},\,\braket{x_i|\psiw}=0,\;\imag (V)=0,
\end{split}
\end{equation}
which we will call the Matsumoto bound. It was also shown in~\cite{matsumoto2002new} that the bound is always saturable using individual measurements (which relaxes the requirement of collective measurements in the pure state case). However, there are no known efficient numerical algorithms to find the solutions $\{\ket{x_i}\}$ of the Mastumoto bound and it is not clear yet whether or not the Mastumoto bound will be useful in finding the optimal QEC protocols in our situation.
In the following, we will start the proof of Theorem 2 by first sketching the proof of the Matsumoto bound and then reformulating it in such a way that the final optimization problem becomes an SDP, which will eventually be incorporated into the QEC protocol optimization procedure.

\subsection{Proof of Theorem 2}

\subsubsection{Proof and reformulation of the Matsumoto bound \texorpdfstring{(\eqref{eq:matsumoto})}{(Eq.(10))}}

According to the Naimark's theorem~\cite{Holevo1982}, for any general measurement $\{M_\ell\}$ on $\mH$ there exists a projective measurement $\{E_\ell\}$ on an extended space $\mH_M$ (where $\mH\subseteq\mH_M$) satisfying $E_\ell E_{\ell'} = \delta_{\ell \ell'}E_\ell$ and $M_\ell=\Pi_\mH E_\ell \Pi_\mH$.
 We now define a set of vectors $\ket{x_i}\in \mH_M$:
\begin{equation}
\label{xdef}
\ket{x_i}=\sum_\ell (\tilde\omega_i(\ell)-\omega_i)E_{\ell}\ket{\psiw}.
\end{equation}
One may see that, thanks to the projective nature of measurements  $\{E_\ell\}$, scalar products of vectors $\ket{x_i}$
yield the covariance matrix of the estimator:
 \begin{equation}
   V_{ij}=\braket{x_i|x_j} = \sum_{\ell,\ell'} \bra{\psiw}(\tilde\omega_i(\ell)-\omega_i)E_{\ell} E_{\ell'}(\tilde\omega_j(\ell')-\omega_j)\ket{\psiw} = \Sigma_{ij}.
 \end{equation}

Now, instead of minimizing over the measurement $\{M_\ell\}$ on $\mH$, we can perform the minimization directly over the vectors $\ket{x_i}\in \mH_M$,
imposing the following constraints:
\begin{equation}
\label{matcond}
\imag(\braket{x_i|x_j})=0, \quad \braket{x_i|\psiw}=0,
\quad 2\real(\braket{x_i|\partial_j|\psiw})=\delta_{ij}.
\end{equation}
These constraints correspond respectively to the projective nature of the measurement $\{E_\ell\}$ and the unbiasedness conditions as given in \eqref{eq:unbiased}.
At this point one may wonder how big the space $\mH_M$ should be (as for a general measurement it might be arbitrary large). However, we can always map $\vspan\{\ket{\psiw},\{\partial_i\ket{\psiw},\ket{x_i}\}_{i=1}^{P}\}\subseteq\mH_M$ isometrically to a $(2P+1)$-dimensional space. Therefore when looking for the bound, under the constraint \eqref{matcond}, it is enough to perform the minimization over $\ket{x_i} \in \vspan\{\ket{\psiw},\partial_1\ket{\psiw},...,\partial_P\ket{\psiw}\}\oplus \mathbb C^P$, which results in equation \eqref{eq:matsumoto}.

Finally, we show that indeed for any set of $\ket{x_i}$ satisfying \eqref{matcond} there exists a proper projective measurement on $\mH\oplus \mathbb C^P$ and a locally unbiased estimator satisfying  \eqref{xdef}, and consequently there exists a corresponding general measurement on $\mH$. To see this, notice that since $\forall_i\braket{\psi_{\boldsymbol{\omega}}|x_i}=0$ and $\forall_{i,j}\braket{x_i|x_j}\in \mathbb{R}$ one may choose an orthonormal basis $\{\ket{b_i}\}$ of $\t{span}\{\ket{\psi_{\boldsymbol{\omega}}},\ket{x_1},\ldots,\ket{x_P}\}$ satisfying:
$\forall_i\braket{\psi_{\boldsymbol{\omega}}|b_i}\in \mathbb{R}\backslash\{0\}$ and $\forall_{i,j}\braket{x_i|b_j}\in \mathbb{R}$. Then one can define a projective measurement:
\begin{equation}
E_\ell=\ket{b_\ell}\bra{b_\ell}\; (\ell=1,\ldots,P+1),\quad
E_0=\openone_{\t{dim}(\mathcal{H}_M)}-\textstyle{\sum}_{\ell=1}^{P+1}\ket{b_\ell}\bra{b_\ell},
\end{equation}
with the corresponding estimator:
\begin{equation}
\tilde\omega_i(\ell)=\frac{\braket{b_\ell|x_i}}{\braket{b_\ell|\psi_{\boldsymbol{\omega}}}}+\omega_i,\;\ell\geq 1,\quad \tilde\omega_i(0)=0,
\end{equation}
which is locally unbiased and satisfies
\begin{equation}
\ket{x_i}=\sum_{\ell=0}^{P+1}(\tilde\omega_i(\ell)-\omega_i)E_\ell\ket{\psi_{\boldsymbol{\omega}}}.
\end{equation}
This proves \eqref{eq:matsumoto}.

Specifically, if $\dim(\mH)\geq 2P+1$ we may choose 
$\vspan\{\ket{\psiw},\partial_1\ket{\psiw},...,\partial_P\ket{\psiw}\}\oplus \mathbb C^P$ as a subspace of $\mH$ and optimize over $\ket{x_i}\in\mH$. In this case we may also reformulate the Matsumoto bound in a slightly different form. First, note that any vectors $\{\ket{x_i}\}$ satisfying \eqref{matcond} need to be linearly independent. Let $\{\ket{c_i}\}_{i=1}^P$ be an orthonormal basis of $\vspan\{\ket{x_1},...,\ket{x_P}\}$, satisfying $\forall_{i,j}\imag \braket{x_i|c_j}=0$ (such a set may be generated using the Gram-Schmidt orthonormalization procedure). The locally unbiased conditions may now be rewritten as:
\begin{equation}
2\real(\braket{x_i|\partial_{j}|\psiw})=
\sum_{k=1}^P2\real(\braket{x_i|c_k}\braket{c_k|\partial_{j}|\psiw})=
\sum_{k=1}^P\braket{x_i|c_k}2\real(\braket{c_k|\partial_{j}|\psiw})=\delta_{ij},
\end{equation}
which (after introducing matrices $\mathcal X_{ki}=\braket{c_k|x_i}$, $\D_{kj}=2\real(\braket{c_k|\partial_{j}|\psiw})$ is equivalent to the matrix equality $\mathcal X^TD=\openone_{\scriptscriptstyle P}$. From $\mathcal X^T \D=\openone_{\scriptscriptstyle P}$ we have $\mathcal X^T=\D^{-1}\Rightarrow \trace(W\cdot V)=\trace(W\cdot \mathcal X^T \mathcal X)=\trace(W\cdot (\D^T\D)^{-1})$, which gives
\begin{equation}
\begin{split}
\label{eq:orthomat}
&\min_{\ket{c_1},...,\ket{c_P}\in\mH}\trace(W\cdot (\D^T\D)^{-1})),\\
&\text{where}~\D_{ij}=2\real\braket{c_i|\partial_j|\psi_\bomega},\quad \text{subject to}~\braket{c_i|c_j}=\delta_{ij}.
\end{split}
\end{equation}
This formulation will be more convenient to use when we will formulate the QEC protocol optimization problem as an SDP.

\subsubsection{Optimizing the error-correction codes}
Now we apply the reformulated Matsumoto bound to our task of identification of the optimal JNT-QEC.
Consider a given input state $\ket{\psi_{\rm in}}$. Let $\mHC$ be any code subspace of $\mH_S\otimes\mH_A$ containing $\ket{\psi_{\rm in}}$ and satisfying the QEC conditions \eqref{eq:QECs}---in order to be in accordance with the reformulated Matsumoto bound, this space may be required to be at least $2P+1$ dimensional,  but as we show in the following it will always be possible to  reduce its dimensionality to $P+1$ effectively. Using QEC, our goal is to preserve an effective unitary evolution in the encoded space and coherently acquire the sensing signal. Therefore, we are effectively dealing with pure state $\ket{\psi_{\bomega}}$, which allows us to utilize \eqref{eq:orthomat}
as a formula for the minimal cost of sensing multiple parameters.

The effective evolution after implementing QEC is given by
\begin{equation}
\ket{\psiw}=\exp\left(-iT\sum_{j=1}^P\omega_j\Pi_{\mHC}\left(G_j\otimes\openone_{\t{dim}(\mathcal{H}_A)}\right)\Pi_{\mHC}\right)\ket{\psi_{\rm in}}.
\end{equation}
 We focus on the estimation around point $\bomega=[0,\ldots,0]$ (which can always be achieved by applying inverse Hamiltonian dynamics~\cite{yuan2016sequential}) and denote $\ket{c_0}=\ket{\psi_{\bomega=0}}$ for notational simplicity. Then for any $\ket{c_i}\in\mHC$ we have $2\real\braket{c_i|\partial_j|\psi_{\bomega=0}}=2T\imag\braket{c_i|(G_j\otimes\openone_{\t{dim}(\mathcal{H}_A)})|c_0}$, and according to \eqref{eq:orthomat} the minimum achievable cost for a fixed code space $\mH_\mC$ is given by:
\begin{equation}
\begin{split}
\label{ourmat}
&\min_{\ket{c_1},...,\ket{c_P}\in
\mHC}
\trace(W\cdot (\D^T\D)^{-1})), \\
&\text{where}~\D_{ij}=2T\imag[\braket{c_i|(G_j\otimes\openone_{\t{dim}(\mathcal{H}_A)})|\psi_\bomega}],\quad\text{subject to}~\braket{c_i|c_j}=\delta_{ij}.
\end{split}
\end{equation}
From the above formulation it is clear that we may always reduce the code space $\mHC$ to $\vspan\{\ket{c_k}\}_{k=0}^{P}$ without increasing the cost. Hence, the problem of optimization over both probes and error-correction protocols is now equivalent to identification of the set $\{\ket{c_k}\}_{k=0}^{P}$ that minimizes the cost with the constraint that $\mHC=\vspan\{\ket{c_k}\}_{k=0}^{P}$ satisfies the QEC conditions.

To solve this problem, it will be convenient to formally extend the Hilbert space $\mH_S \otimes \mH_A$ by tensoring it with a $(P+1)$-dimensional reference space $\mHR = \t{span}\{\ket{0},\dots,\ket{P}\}$ (see \figcref{fig:diagram}).
This reference space will be representing the effective
evolution of the probe state that happens within the code space and it will allow us to
encode QEC conditions in a compact and numerically friendly way.

First, we introduce a matrix $Q \in {\cal L}(\mHR\otimes \mH_S)$ that represents a code
\begin{equation}
\label{eq:qdef}
Q = \trace_{\mH_A}
\begin{pmatrix}
\begin{bmatrix}
\ket{c_0}\\
\vdots\\
\ket{c_P}
\end{bmatrix}
\begin{bmatrix}
\bra{c_0} &
\cdots &
\bra{c_P}
\end{bmatrix}
\end{pmatrix},
\end{equation}
which, more formally, will be written as $Q = \trace_{\mH_A}(\sum_{k,l=0}^{P}\ket{k}\bra{l}_\mHR\otimes\ket{c_k}\bra{c_l}_{\mH_S\otimes\mH_A})$.
This matrix is proportional to the reduced density matrix of the maximum entangled state between $\mHR$ and $\mH_\mC$.
By its construction $Q \geq 0$ and contains all relevant information on the code states in $\mH_\mC$.

Next, we introduce effective generators $G^{\eff}_i$ acting on $\mHR$ so that they represent properly the action of the physical generators on the code space $[G^{\eff}_i]_{kl}=\braket{c_k|G_{i}\otimes \id_{\t{dim}(\mathcal{H}_A)}|c_l}$.
The effective evolution generators are related with the $Q$ matrix via:
\begin{equation}
\label{eq:Ceff}
(G^\eff_i)^T = \trace_{\mH_S} \left[Q (\id_{\scriptscriptstyle P+1} \otimes G_{i})\right]
\quad i=1,\ldots,P.
\end{equation}
Note that the identity operator here acts on the reference space $\mHR$, and \emph{not} on the ancillary space $\mH_A$. Taking into account the orthonormality of $\ket{c_k}$ and the QEC condition \eqref{eq:QECs}, we obtain the following constraints on $Q$
\begin{equation}
\label{eq:CQEC}
\trace_{\mH_S}(Q)=\id_{\scriptscriptstyle P+1}, \;\forall_{S_i\in\mS}~\trace_{\mH_S}\left[Q (\id_{\scriptscriptstyle P+1} \otimes S_i)\right]\propto \id_{\scriptscriptstyle P+1}.
\end{equation}
Let $\id_d/\sqrt{d}$, $\{G_{i}\}_{i=1}^P$, $\{S_{i}\}_{i=P+1}^{P'}$, $\{R_{i}\}_{i=P'+1}^{d^2-1}$ form an orthonormal basis of Hermitian operators in ${\cal L}(\mH_S)$ such that $\mS = \mathrm{span}_{\mathbb{R}}\{\id_d, {(S_{i})}_{i=P+1}^{P'}\}$.
Any non-negative $Q$ satisfying \eqsref{eq:Ceff}{eq:CQEC} has the following form:
\begin{equation}
\label{eq:C}
Q = \id_{\scriptscriptstyle P+1} \otimes \frac{\id_d}{d} + \sum_{i=1}^{P} (G^\eff_i)^T \otimes G_{i}
 + \sum_{i=P+1}^{P'} \nu_i \id_{\scriptscriptstyle P+1} \otimes S_i + \sum_{i=P^\prime +1}^{d^2-1} B_i \otimes R_i\geq 0,
\end{equation}
where $\nu_i\in\mathbb R$ and $B_i$ are Hermitian. Conversely,
for any nonnegative defined $Q\geq 0$, we can consider its purification
$\ket{Q} \in \mHR \otimes \mH_S \otimes \mH_A$, which when written as $\ket{Q} = \sum_{k=0}^P \ket{k}_{\mHR} \otimes \ket{c_k}_{\mH_S \otimes \mH_A}$ yields the code states $\ket{c_k}$. Note that it implies that the rank of $Q$ corresponding to the dimension of the ancillary space. It is always sufficient to assume the dimension of the
ancillary space to be $\dim \mH_A =(P+1)d$. Therefore $\{G^{\eff}_i\}$ is an achievable set of effective generators in ${\cal L}(\mHR)$ (satisfying the QEC condition) if and only if there exist such $\nu_i\in\mathbb R$ and $B_i$, for which $Q \geq 0$.

Finally, in order to have an explicit dependence of the cost on the total time parameter $T$, we introduce a matrix $\A=\frac{1}{2T}\D$, i.e. $\A_{ij}=\imag[\braket{c_i|G_j\otimes\openone_{\t{dim}(\mathcal{H}_A)}|c_0}]=\imag[G^\eff_j]_{i0}$, and we end up with:
\begin{equation}
\begin{split}
\label{almostend}
&\frac{1}{4T^2} \min_{G^\eff_i,B_i,\nu_i}\trace \left(\weight (\A^T\A)^{-1}\right),\quad\text{where~}\A_{ij}=\imag[G^\eff_j]_{i0},\\
&\text{subject to}~ \id_{\scriptscriptstyle P+1} \otimes \frac{\id_d}{d}  +  \sum_{i=1}^{P}  (G^\eff_i)^T  \otimes  G_{i}  +  \sum_{i=P+1}^{P'}  \nu_i \id_{\scriptscriptstyle P+1}  \otimes  S_i +  \sum_{i=P^\prime +1}^{d^2-1}  B_i  \otimes  R_i  \geq  0.
\end{split}
\end{equation}

\subsubsection{Reduction to an SDP}
In order to reformulate \eqref{almostend} as an SDP, we first show that we may assume without loss of generality that $\A\sqrt{\weight^{-1}}\geq 0$. Note that for any full rank matrix $\A$, the polar decomposition theorem implies that there always exists an orthonormal matrix $O$ such that $O\A\sqrt{\weight^{-1}}\geq 0$. Next, as $\A_{ij}=\imag[\braket{i|G^\eff_j|0}]$, multiplication $\A$ by $O$ is equivalent to rotating the base in the reference space $\mHR$. Since, according to \eqref{eq:qdef}
 such a rotation cannot change the non-negativity of $Q$ and at the same time it does not affect the figure of merit $\trace\left(\weight(\A^T\A)^{-1}\right)$, the statement is proven. To put \eqref{almostend} in the form of an SDP, we introduce a positive matrix $K \in \bR^{P\times P}$ and a positive real number $w$. Now, using the following two relations,
\begin{equation}
\begin{bmatrix}
K&\openone_{\scriptscriptstyle P}\\
\openone_{\scriptscriptstyle P}&\A\sqrt{\weight^{-1}}\\
\end{bmatrix}\geq 0~~\Leftrightarrow~~ K\geq (\A\sqrt{\weight^{-1}})^{-1},
\end{equation}
\begin{equation}
\begin{bmatrix}
w\openone_{\scriptscriptstyle P}&K\\
K&\openone_{\scriptscriptstyle P}\\
\end{bmatrix}\geq 0~~\Leftrightarrow~~ w \openone_{\scriptscriptstyle P}\geq K^2,
\end{equation}
we see that $P \min w = \min \trace(K^2) = \min \trace\left(W(\A^T\A)^{-1}\right)$ in \eqref{eq:thm2}, making it equivalent to \eqref{almostend}. Hence the problem takes the form of an SDP.

\subsection{Discussion}
It should be remarked that JNT-QEC do not contain SEP-QEC as a subclass.  In SEP-QEC, unlike in JNT-QEC, the noises are not fully corrected in the entire space, and the decoherence is only avoided by choosing a properly mixed state input. In general one could combine both these approaches in a unified framework by dividing the set of all parameters into smaller subsets and then applying JNT-QEC for each of these subset separately---in this approach the SEP-QEC case would correspond to the situation where JNT-QEC optimization is applied to single parameter subsets. Such an optimization is in principle doable, but will involve much large numerical effort and it is not clear that it will lead to better protocols.

It is also worth noting that, apart from the improved metrological performance provided by QEC protocols when dealing with noisy systems, the above algorithm is also applicable in the noiseless scenario when $\mS = \t{span}_{\mathbb{R}}\{\openone_d\}$. In such a situation \emph{no QEC is required} (for simplicity we may still use $\mHC$ for $\vspan\{\ket{c_k}\}_{k=0}^{P}$, but no recovery operation or projection $\Pi_{\mHC}$ is needed during evolution), but the condition $\bomega=[0,\ldots,0]$ (which is achievable by applying inverse Hamiltonian dynamics~\cite{yuan2016sequential}) is still required,  as otherwise the derivatives of the state may not scale linearly with $T$.
In such situations, the solution of JNT-QEC 
yields an optimally ancilla-assisted sensing protcol under arbitrary system dynamics (Hamiltonitans) that resolves the potential incompatibility issues between sensing of different parameters. 
It should be stressed that our approach is universal and unlike existing approaches \cite{baumgratz2016quantum,yuan2016sequential,Kura2018} does not assume any specific structures of the Hamiltonians. 

\section{Advantages of JNT-QEC over SEP-QEC}
\label{sec:JNTSEP}
In this section we investigate the potential advantages JNT-QEC provide compared with SEP-QEC.
The characteristic feature of the SEP-QEC is the division resources so that each part of the resources is used to measure a given parameter independently of the others. This is  reflected in the form of the input probe state $\rho_{\t{in}}=\frac{1}{P}\sum_{i=1}^P\ket{\psi_i}\bra{\psi_i}$. As a consequence we effectively measure each parameter only once in every $P$ repetitions of an experiment (corresponding to the $1/P$ factor in the $\rho_{\t{in}}$). Therefore for a fixed total number of measurements, the uncertainty of estimating a given parameter will grow proportionally to $P$. Contrastingly, in JNT-QEC  there is a chance to avoid the division of resources and use a single pure state and a single measurement to estimate all of them simultaneously.
If there existed a code space, a state and a measurement which were all simultaneously optimal for all parameters, 
then the decrease the cost by a factor $P$ would be achievable. This would be the largest possible advantage offered by JNT-QEC. 
In general we can write (see \appref{sec:jntsep} for more formal derivation):
\begin{equation}
\label{eq:jntpsep}
\cost_\JNT\geq\tfrac{1}{P}\cost_\SEP.
\end{equation}
This inequality may be saturated only in special examples. 

It is important to note that Theorem 2 gives us a recipe for identification of the optimal JNT-QEC, while the explicit construction of SEP-QEC discussed in the proof of Theorem 1 was only aimed at demonstrating the possibility of the HS and hence the protocol was not optimized. Therefore, to make a fair comparison between the two approaches, we need to compare the performance of the best JNT-QEC with the best SEP-QEC.

Below we present a way to find the optimal  SEP-QEC, i.e. the protocol for which all parameters are measured independently on mutually orthogonal subspaces. We will also provide a useful lower bound for the minimal cost achievable in such protocols.

\subsection{Optimization of SEP-QEC}
Adapting the results from~\cite{zhou2018achieving} (see the section named “Geometrical picture”), we can 
infer that for a given set of $\{G_i\}$, the minimal variance achievable in estimation of a single parameter $\omega_i$ in SEP-QEC is given by:
\begin{equation}
\label{sepopt}
\Delta^2\tilde\omega_i=\frac{1}{F_i},\quad F_i=4T^2\min_{\widetilde G_{i\parallel}\in\mS\oplus\mathrm{span}_{\mathbb{R}}\{ G_j\}_{j\neq i}}\|G_i-\widetilde G_{i\parallel}\|,
\end{equation}
where $\|\cdot\|$ denotes operator norm. Let $\ket{\psi_i}$ be the optimal state for measuring $\omega_i$. Therefore, using $\rho_{\rm in}=\sum^P_{i=1}p_i\ket{\psi_i}\bra{\psi_i}$ (where in the naive approach we would set all $p_i=\frac{1}{P}$)  leads to the total cost $\cost_{\SEP}=\sum_{i=1}^P \frac{1}{p_i}\frac{\weight_{ii}}{F_i}$. After optimization over $p_i$ is performed (keeping in mind that $\sum_{i=1}^Pp_i=1$, $p_i \geq 0$) we get $\cost_{\SEP}=\left(\sum_{i=1}^P\sqrt{\frac{\weight_{ii}}{F_i}}\right)^2$.

Next, we may still improve the peromance of  SEP-QEC by choosing a different QEC code based on a new set of generators obtained via a linear transformation on the parameters $A \in \bR^{P\times P}$:
\begin{equation} 
\bomega' =\bomega A^{-1}, \quad \boldsymbol{G}' =A \boldsymbol{G},\quad \weight' = A\weight A^{T},
 \end{equation}
so that the cost function remains unchanged.
 
 Note also, that rescaling any generator by constant factor has no impact on the results, so we may restrict to linear transformations satisfying
$(AA^T)_{ii}=1$.
Therefore the cost for the optimal SEP-QEC is given by:
\begin{equation}
\begin{split}
\label{optsep}
\cost_{\rm SEP}&=\min_{A: \forall_{i} (AA^T)_{ii}=1}\left(\sum_{i=1}^P\sqrt{\frac{(AWA^T)_{ii}}{F_i(A)}}\right)^2,\\
&{\rm where}\quad F_i(A)=4T^2\min_{\widetilde G_{i\parallel}\in\mS\oplus\mathrm{span}_{\mathbb{R}}\{ (A\boldsymbol G)_j\}_{j\neq i}}\|(A\boldsymbol G)_i-\widetilde G_{i\parallel}\|.
\end{split}
\end{equation}
We want to stress that introducing this (relatively complicated) procedure of SEP-QEC optimization is necessary in order to distniguish the cases where the true advantage is offered by the joint multi-parameter approach compared with the situation where the advatage is only apparent and results from a suboptimal choice of separate protocols.

Formula \eqref{optsep} is rather complicated and hard to compute in general. However, in  case $\weight=\openone$ the following lower bound is valid:
\begin{equation}
\label{sepbound}
\cost_\SEP\geq \min_{a_i:\sum a_i^2=1}\frac{P^2}{T^2(\lambda^*_+-\lambda^*_-)^2},\quad G^*=\sum_{i=1}^Pa_iG_i,
\end{equation}
where $\lambda^*_\pm$ are extreme eigenvalues of $G^*$.


 Since the reasoning leading to \eqref{eq:jntpsep} holds for any  set of generators, not necessary the ones optimal for SEP-QEC, therefore,  the optimal joint-estimation cost, in case $\weight=\openone$,  may be bounded as (see \appref{sec:jntsep} for a formal derivation):
\begin{equation}
\label{jntsepbound}
\cost_{\JNT}\geq \min_{i}\frac{P}{T^2(\lambda_{i+}-\lambda_{i-})^2},\quad {\rm for\,any\,set\,of}\,G_i
\end{equation}
where $\lambda_{i\pm}$ are extreme eigenvalues of $G_i$.

In order to provide the reader with some intuition on the concept presented above, we discuss below two extreme cases 
illustrating the apparent and maximal advantage of joint-estimation protocols over separate ones. These examples deal with noiseless scenarios where JNT-QEC is always optimal and are aimed to prepare the reader for more physical noisy examples
discussed in \secref{sec:examples}. Still, in order not to introduce additional abbreviations we will still use the acronyms JNT-QEC, SEP-QEC to 
describe separate and joint estimation schemes, even though the role of QEC in these examples is trivial.

\subsection{Example: Apparent advantage of JNT-QEC}
Consider a noiseless physical system comprising $P=2^r$ qubits (for technical reasons we require the number of qubits to be the integer power of $2$). Each qubit, regarded as a spin 1/2, senses the $z$-component of a local magnetic field which is assumed to be independent for different qubit locations. The corresponding sensing Hamiltonian for this system reads:
\begin{equation}
H=\sum_{i=1}^P\omega_i G_i,\quad
G_i=\sigma_z^{(i)}=\openone^{\otimes (i-1)}\otimes\sigma_z\otimes \openone^{\otimes (P-i)}
\end{equation}
and the cost matrix is assumed to be $W = \openone$. 
The minimum variance of estimating each parameter independently is lower bounded by $\Delta^2\tilde\omega_i\geq \frac{1}{T^2(\lambda_+-\lambda_-)^2}$. As the maximum and the minimum eigenvalues of each $G_i$ are $\lambda_{i\pm}=\pm 1$, the 
 minimal variance reads $\Delta^2\tilde\omega_i=\frac{1}{4T^2}$.
According to \eqref{jntsepbound}) this implies that the optimal JNT-QEC cost is lower bounded by $\frac{P}{4T^2}$. Moreover, since
 the state $\ket{\psi}=\tfrac{1}{\sqrt{2^P}}\left(\ket{+}+\ket{-}\right)^{\otimes P}$ is simultaneously optimal for all the parameters (and 
 the tensor structure guarantees no measurement incompatibility issue), this bound is saturable and hence
 \begin{equation}
 \cost_{\JNT} = \frac{P}{4T^2}. 
 \end{equation} 
  On the other hand, the estimation strategy where each parameter $\omega_i$ is estimated separately leads to the cost equal $\frac{P^2}{4T^2}$. Therefore, one may naively think that it is an example of superiority JNT-QEC over SEP-QEC.

However, the separate protocol may be significantly improved here, by estimating different combinations of the parameters. 
According to \eqref{sepbound}, we should be looking for a proper linear combination of $G_i$ with the biggest difference of extreme eigenvalues. The most obvious choice is $G^*=\frac{1}{\sqrt{P}}\sum_{i=1}^PG_i$, for which $\lambda^*_{\pm}=\sqrt{P}$, and therefore
 $\cost_\SEP\geq \frac{P}{4T^2}$, which is exactly equal to $\cost_{\JNT}$!  Below we show, that this bound may be saturated.

Let the matrix $A$ defining the transformation to the new set of parameters be 
proportional to a Hadamard matrix of size $P\times P$ (i.e. square matrix whose entries are either $+1$ or $-1$ and whose rows are mutually orthogonal) defined as follows:
\begin{equation}
A_{i j} = \frac{1}{\sqrt{P}}\prod_{k=0}^{r-1} (-1)^{i_k j_k},
\end{equation}
where the indices of the matrix are written using the binary representation, 
$i = \sum_{k=0}^{r-1} i_k 2^k$, where $i_k$ represent the binary digit of $i$ at position $k$. 
%
Then for $\boldsymbol{G}'=A\boldsymbol{G}$ we have:
\begin{equation}
\lambda'_{i\pm}=\pm\sqrt{P},\quad\ket{\lambda'_{i\pm}}=\bigotimes_{j=1}^P\ket{\pm A_{ij}},\quad \forall_{j\neq i}\braket{\lambda'_{i\pm}|G'_j|\lambda'_{i\pm}}=\braket{\lambda'_{i\pm}|G'_j|\lambda'_{i\mp}}=0.
\end{equation}
As a result every $\omega_i'$ may be measured separately on subspace $\vspan\{\ket{\lambda'_{i+}},\ket{\lambda'_{i-}}\}$ with precision $\frac{1}{4PT^2}$. Therefore, for the initial sate $\rho_{\rm in}=\frac{1}{P}\sum_{i=1}^P\ket{\psi_i}\bra{\psi_i}$ with $\ket{\psi_i}=\tfrac{1}{\sqrt{2}}(\ket{\lambda_{i+}}+\ket{\lambda_{i-}})$ we get 
\begin{equation}
\cost_{\SEP}=\frac{P}{4T^2},
 \end{equation}
 which is exactly the same as the optimal $\cost_{\JNT}$. We see that the apparent advantage of a JNT-QEC disappears once a proper  combinations of parameters are estimated in the SEP-QEC.

\subsection{Example: Maximal advantage of JNT-QEC}
\label{sec:exmaxjnt}
As a contrasting example, here we discuss a situation where the advantage of JNT-QEC is genuine and is the maximal possible.
This example also provides an intuition, what relation between the generators $G_i$ is responsible for this. 

Consider a noiseless system $\mH_S=\vspan\{\ket{i}\}_{i=0}^P$ with Hamiltonian:
\begin{equation}
H=\sum_{i=1}^P\omega_iG_i,\quad G_i=\frac{1}{\sqrt{2}}(\ket{0}\bra{i}+\ket{i}\bra{0}).
\end{equation}
We focus on the estimation around point $\bomega=[0,\ldots,0]$ with the cost matrix $\weight=\openone$. As all generators are orthonormal $\trace(G_i G_j)=\delta_{ij}$, the difference between  extreme eigenvalues 
of any normalized combination of them, such as $G^*$ in \eqref{sepbound}, is at most equal to $\sqrt{2}$. Therefore, the optimal SEP-QEC cost is bounded by 
\begin{equation}
\cost_\SEP\geq \frac{P^2}{2T^2}.
\end{equation}

Importantly, the state $\ket{\psi_{\rm in}}=\ket{0}$ is simultaneously optimal for measuring all the parameters.  
There is also no  measurement incompatibility problem as the optimal measurement for all the parameters is the measurement in the basis $\ket{i}$.
Therefore, 
\begin{equation}
\cost_\JNT=\frac{P}{2T^2}
 \end{equation}
is achievable, and hence the there is a factor $P$ decrease in the cost in case of JNT-QEC compared with the optimal SEP-QEC protocols.

\section{Examples}
\label{sec:examples}
Here we provide representative examples of a large class of multi-parameter estimation models, where there are unavoidable tradeoffs 
in determining the optimal states and measurements that arise due to the multi-parameter nature of the problem. Our methods are useful in identifying optimal strategies in all such models, provided the structure of noise admits the achievability of the HS via application of the most general QEC schemes.
By the construction of our algorithm (Theorem 2), we have the guarantee that the solutions found are the optimal ones. For an interested reader, a broader  discussion and generalizations of the results presented in this section, including proofs and analytical constructions of the codes
are provided in \appref{sec:qubit}, \appref{sec:2qubits}, \appref{sec:sud}.

\subsection{Single qubit case}
Consider first the simplest single-qubit case with $d=2$. The HS is achievable via QEC only in the case of single-rank Pauli noise (specified by a single Hermitian Lindbladian $L$)~\cite{sekatski2017quantum}. Without loss of generality we can set $L=\sigma_z$ (the Pauli-Z matrix). Since $\mathcal{S}=\t{span}\{ \openone, \sigma_z\}$, at most two parameters may be estimated in a qubit system with the HS (as $\dim(\mS^{\perp})=2$). However, it turns out that when the multi-parameter HNLS condition is met there is no benefit in performing the more sophisticated JNT-QEC compared to SEP-QEC, which is shown analytically in \appref{sec:qubit}.

\subsection{Two qubits in a magnetic field}
In order to appreciate the superiority of JNT-QEC over SEP-QEC, let us consider a two-qubit model which is a multi-parameter generalization of the one from~\cite{layden2018spatial}. Consider two localized qubits, coupled to a magnetic field, which is constant in both time and space, apart from some small fluctuations in the $z$ direction. These fluctuations are assumed to be uncorrelated in time, but maximally anticorrelated in space (for the two qubits they have always opposite signs). Such a system may be effectively described by \eqref{eq:evol} with $H=\frac{1}{2}\sum_{i=1}^2\bomega \cdot\bsigma^{(i)}$ (where $\bsigma^{(i)}=[\sigma_x^{(i)},\sigma_y^{(i)},\sigma_z^{(i)}]$ acts on the $i^{\text{th}}$ atom) and a single Lindblad operator $L=\sqrt{2\gamma}(\sigma_z^{(1)}-\sigma_z^{(2)})$. It can be shown that the minimal  cost when each parameter is estimated with the optimal individual parameter strategy is $\Delta^2\tilde\omega_{x,y,z}=\frac{1}{4T^2}$. At the same time, in accordance with the discussion presented in Sec.~\ref{sec:JNTSEP}, the best precision achievable using SEP-QEC for the standard cost matrix $W=\openone$, is $\cost_{\t{SEP}} = \frac{P^2}{4T^2} = \frac{9}{4T^2}$---see
\appref{sec:2qubits} for the formal derivation of the above formulas.

Below we present the result of numerical optimization of the JNT-QEC approach (found by the algorithm presented in Theorem~2 and reconstructed to its analytical form). We will use the standard Bell states notation:
\begin{equation}
\begin{split}
\uup=\frac{1}{\sqrt{2}}(\ket{\up\up}+\ket{\down\down}),\quad
\uum=\frac{1}{\sqrt{2}}(\ket{\up\up}-\ket{\down\down}),\\
\udp=\frac{1}{\sqrt{2}}(\ket{\up\down}+\ket{\down\up}),\quad
\udm=\frac{1}{\sqrt{2}}(\ket{\up\down}-\ket{\down\up}).
\end{split}
\end{equation}
Entanglement with ancilla will be abbreviated in the subscript $\ket{\psi}\otimes\ket{i}_A\equiv\ket{\psi}_i$. Using the numerical algorithm we have found out, that the optimal code space has the form
\begin{equation}
\begin{split}
\ket{c_0}&=-\cos(\varphi)\uup_1+\frac{i}{\sqrt{2}}\sin(\varphi)(\uup_2+\uum_3),\\
\ket{c_1}&=-i\sin(\varphi)\uup_1-\cos(\varphi)\udp_2,\\
\ket{c_2}&=-\sin(\varphi)\uum_1-i\cos(\varphi)\udp_3,\\
\ket{c_3}&=-\frac{1}{\sqrt{2}}\sin(\varphi)(\uum_2+\uup_3)+\cos(\varphi)\udp_4,\\
\end{split}
\end{equation}
where the input state is $\ket{\psiwn}=\ket{c_0}$. Note that the presence of the last term in $\ket{c_3}$ (entangled with $\ket{4}_A$) is necessary to satisfy the QEC conditions. The value of $\varphi$ can be found  analytically and the minimal total cost of estimation $\cost$ is achieved for:
\begin{equation}
\cos(\varphi)=\sqrt{\frac{\sqrt{7+4\sqrt{2}}-3}{4\sqrt{2}-2}}\approx 0.39,
\end{equation}
while the corresponding optimal cost is:
\begin{equation}
\cost_\JNT\approx \frac{5.31}{4T^2}.
\end{equation}
As we see a significant improvement has been achieved here compared to SEP-QEC.

\begin{figure}[t]
\center
\includegraphics[width=0.55\columnwidth]{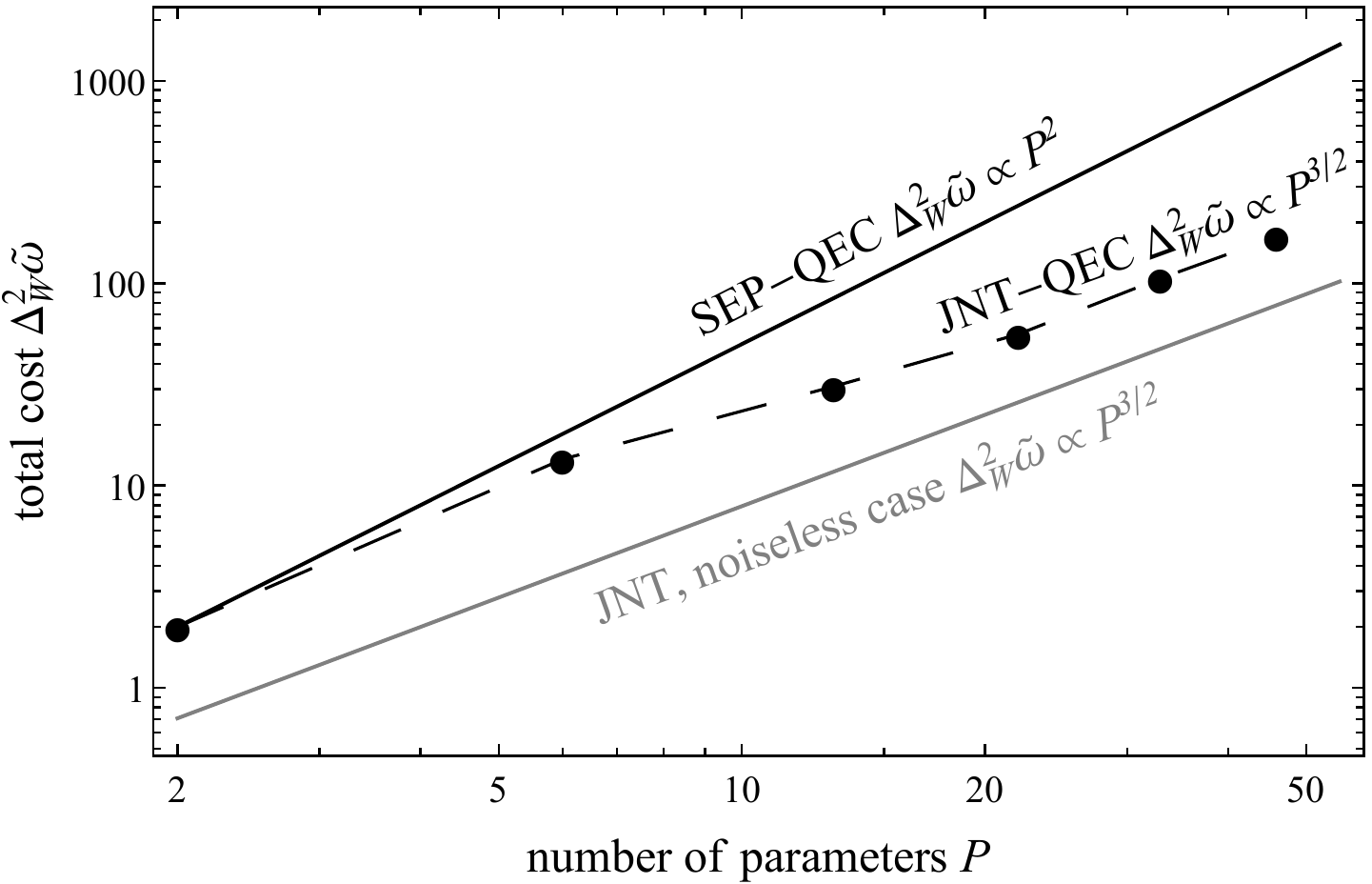}
\caption{
Numerical results for the optimal JNT-QEC strategy for estimating the $SU(d)$ generators under the noise $J_z$ (black points joined by dashed line), in contrast with the optimal precision asymptotically achievable by SEP-QEC $\cost=\frac{P^2}{2T^2}$ (black solid line) and the lower bound $\cost=\frac{P^{3/2}}{4T^2}$ asymptotically achievable only in the noiseless case (solid gray line).}
\label{fig:sud}
\end{figure}

\subsection{\texorpdfstring{$SU(d)$}{SU(d)} generators' estimation}
\label{sec:exsud}
Finally let us consider an example of estimating parameters associated with $SU(d)$ generators  which shows an asymptotic advantage (with the number of parameters) of the JNT-QEC over the SEP-QEC protocol. Let $\mH_S$ be $d$-dimensional Hilbert space, which for a more intuitive notation may be regarded as the one associated with a spin-$j$ particle (where $d=2j+1$).

First, let us recall the noiseless case, where the Hamiltonian $H=\sum_{i=1}^{d^2-1}\omega_i G_i$ is composed of all $P=d^2-1$ $SU(d)$ generators. The generators together with the identity $\{\frac{1}{\sqrt{d}}\openone,G_1,\ldots,G_{d^2-1}\}$ form an orthonormal basis of Hermitian operators on $\mH_S$. 
We focus on the estimation around point $\bomega=[0,\ldots,0]$ with the cost matrix $\weight=\openone$. 
Since all $G_i$ are orthonormal the maximum spread between their maximum and minimal eigenvalues as well as any normalized combination of them is $\sqrt{2}$. Using the same line of reasoning as presented in \secref{sec:exmaxjnt} we conclude that 
\begin{equation}
\cost_\SEP\geq \frac{P^2}{2T^2}=\frac{(d^2-1)^2}{2T^2}.
\end{equation}

The bound for $\cost_\JNT$ may be derived analytically. To achieve that, we use the following chain of inequalities involving the QFI matrix:
\begin{equation}
\label{genin}
\sum_{i=1}^{d^2-1}\Delta^2\tilde\omega_i\geq\sum_{i=1}^{d^2-1}(\fisher^{-1})_{ii}\geq\sum_{i=1}^{d^2-1}\frac{1}{\fisher_{ii}}\geq\frac{(d^2-1)^2}{\sum_{i=1}^{d^2-1}\fisher_{ii}},
\end{equation}
where the first one is the CR inequality and the rest are general algebraic properties of positive semidefinite matrices. What remains to be done is to derive a proper bound for the trace of the QFI matrix.  For any input state $\ket{\psi_{\rm in}}\in\mH_S\otimes\mH_A$ we have:
\begin{equation}
\fisher_{ii}=4T^2\left(\braket{\psiw|G_i^2\otimes\openone|\psiw}-\braket{\psiw|G_i\otimes\openone|\psiw}^2\right)  \leq 4T^2\braket{\psiw| G_i^2\otimes\openone|\psiw}.
\end{equation}
Taking into account the normalization of $G_i$ and noting that $\sum_{i=1}^{d^2-1}G_i^2$ is the Casimir operator of the $SU(d)$ algebra, and 
as such is proportional to the identity, we get that  $\sum_{i=1}^{d^2-1}G_i^2=\frac{d^2-1}{d}\openone$. Therefore:
\begin{equation}
\sum_{i=1}^{d^2-1}\fisher_{ii}\leq 4T^2\braket{\psiw|\sum_{i=1}^{d^2-1}G_i^2\otimes\openone|\psiw}
=4T^2\frac{d^2-1}{d}.
\end{equation}
After substituting the above to \eqref{genin} we get
\begin{equation}
\cost_\JNT = \sum_{i=1}^{d^2-1}\Delta^2\tilde\omega_i\geq\frac{d(d^2-1)}{4T^2}\approx \frac{P^{3/2}}{4T^2}.
\end{equation}

The example of a state which saturates the above bound is $\ket{\psi_{\rm in}}=\frac{1}{\sqrt{d}}\sum_{k=1}^d\ket{k}_S\otimes\ket{k}_A$\cite{Kura2018,yuan2016sequential}. For such a state the QFI matrix is given by $\fisher_{ij}=\delta_{ij}\frac{4T^2}{d}$, so the second and the third inequalities in \eqref{genin} become equalities. As $\imag(\braket{\psiw|\Lambda_i\Lambda_j|\psiw})\propto\braket{\psiw|[G_i\otimes\openone,G_j\otimes\openone]|\psiw}=0$, the first one (the CR bound) is saturable as well. Note that the role of ancillae here is to make optimal measurements with respect to different parameters compatible.
As a result, we see that optimal $\cost_\JNT$ is approximately $2\sqrt{P}$ times smaller than $\cost_\SEP$.

Let us now consider a noisy version with a single Lindblad operator $J_z=\sum_{k=-j}^j k\ket{k}\bra{k}$. From Theorem~1 we know that only parameters associated with generators $G_i\notin\vspan_{\mathbb R}\{\id,J_z,J_z^2\}$ may be estimated with the HS. Therefore we consider the Hamiltonian $H=\sum_{i=1}^{P}\omega_i G_i$ composed of $P=d^2-3$ of $SU(d)$ generators orthogonal to $\vspan_{\mathbb R}\{\id,J_z,J_z^2\}$), with the standard cost matrix $W=\openone$ (such cost makes the problem independent on choosing peculiar set of $\{G_i\}$). In \figref{fig:sud}, we present numerical results for such a problem, and we observe a significant advantage over the SEP-QEC protocol as well as strong indication of the asymptotic $P^{3/2}$ scaling identical to the noiseless case.
Even though the optimal JNT-QEC code cannot be written down analytically in a concise way, in \appref{sec:sud} we provide an analytical suboptimal construction achieving the $P^{3/2}$ scaling, supporting the numerical results.
The scaling advantage we prove here is not trivial because there are no decoherence-free subspaces in the system.


\section{Conclusions and outlook}

We have generalized previous results on single-parameter error-corrected metrology to the multi-parameter scheme, obtaining a necessary and sufficient condition (HNLS) for the achievability of the HS. In case of scenarios when HNLS is satisfied, we developed
 an efficient numerical algorithm (formulated as an SDP) to find the optimal QEC protocol, including the optimal input state, QEC codes and measurements. Our algorithm is applicable to arbitrary system dynamics as long as the HS is achievable (including noiseless cases), 
  which contrasts previous works where special dynamics of quantum system is assumed~\cite{baumgratz2016quantum,yuan2016sequential,Kura2018} or 
  the estimation is performed on fixed quantum states~\cite{albarelli2019evaluating}.
However, it still remains open in which situations the requirement of noiseless ancillae could be removed~\cite{layden2018spatial,layden2018ancilla} and whether QEC is still helpful in the case when HNLS is violated, as in the single-parameter case~\cite{zhou2020optimal,zhou2020theory}.

We also remark that our way of formulating the Knill-Laflamme conditions~\cite{knill1997theory} as a positive semidefinite constraint~\eqref{eq:C} is novel and may have applications beyond error-corrected quantum metrology.

Finally, we note that in this paper we have followed the frequentist estimation approach, and in principle more stringent HS bounds might 
be derived when following the Bayesian approach as was demonstrated recently in the single parameter case \cite{Gorecki2020}. 

\section*{Acknowledgements}
We thank Mengzhen Zhang, John Preskill and Francesco Albarelli for inspiring discussions and Joseph Renes for his input in formulating the final form of the SDP. WG and RDD acknowledge support from the  National Science Center (Poland) grant No. 2016/22/E/ST2/00559.
SZ and LJ acknowledge support from the ARL-CDQI (W911NF-15-2-0067, W911NF-18-2-0237), ARO (W911NF-18-1-0020, W911NF-18-1-0212), ARO MURI (W911NF-16-1-0349), AFOSR MURI (FA9550-14-1-0052, FA9550-15-1-0015), DOE (DE-SC0019406), NSF (EFMA-1640959), and the Packard Foundation (2013-39273).

\onecolumn\newpage
\appendix

\renewcommand{\thesection}{\Alph{section}}
\setcounter{theorem}{0}
\renewcommand{\thetheorem}{A\arabic{theorem}}


\section{The proof of $\cost_\JNT\geq \frac{1}{P}\cost_\SEP$}
\label{sec:jntsep}
In the reasoning presented in the beginning of \secref{sec:JNTSEP} we have assumed that, when estimating a single parameter, the minimal variance of the estimator $\Delta^2\tilde\omega_i$ achievable in SEP-QEC is no bigger than the one achievable in JNT-QEC (even if we focus only on this particular parameter). This statement is not so obvious as both protocols impose different constraints on the code space:
\begin{equation}
\begin{split}
{\textrm{JNT-QEC}}&:\,\forall_{j}\Pi_{\mH_C}(G_j\otimes\openone)\Pi_{\mH_C}\not\propto \Pi_{\mH_C}, \\
{\textrm{SEP-QEC}}&:\,\forall_{j\neq i}\Pi_{\mH_{C_i}}(G_j\otimes\openone)\Pi_{\mH_{C_i}}\propto \Pi_{\mH_{C_i}}, \quad
\Pi_{\mH_{C_i}}(G_i\otimes\openone)\Pi_{\mH_{C_i}}\not\propto \Pi_{\mH_{C_i}}, 
\end{split}
\end{equation}
where none of them is weaker or stronger than the other one.

Therefore, here we show directly that having JNT-QEC with accuracy $\cost_\JNT$, one may always construct SEP-QEC giving $\cost_\SEP=P\cdot \cost_\JNT$. Let $\mH_\mC$ be the code space used in JNT-QEC. Without loss of generality let us assume that $\weight$ is diagonal (otherwise we apply a transformation of parameters that diagonalizes $\weight$) and $\dim(\mH_\mC)\geq 2P+1$ (otherwise we trivially extend it using an additional ancilla). From the Matsumoto bound  (see \eqref{eq:matsumoto}) the optimal cost is given by:
\begin{equation}
\label{aaaa}
\begin{split}
&\cost_\JNT=\min_{\{\ket{x_i}\}}\trace(\weight V),\quad {\textrm{where}}~~V_{ij}=\braket{x_i|x_j},\\
&\text{for}~\ket{x_i} \in \mH_\mC,\\
&\text{subject~to}~2\real (\braket{x_i|\partial_j|\psi_\bomega})=\delta_{ij},\,\braket{x_i|\psiw}=0,\;\imag (V)=0,
\end{split}
\end{equation}
where for $\omega=0$ we have $\partial_j\ket{\psiw}\big|_{\omega=0}=-iT G_j^\mC\ket{\psiwn}$. We define $\mH_{\mC_i}=\vspan_{\mathbb R}\{ \ket{c_0^i}, \ket{c_1^i}\}$ with:
\begin{equation}
\ket{c_0^i}=\tfrac{1}{\sqrt{2}}\left(\ket{\psiwn}\ket{0}_{A_i}+\tfrac{1}{\sqrt{\braket{x_i|x_i}}}\ket{x_i}\ket{1}_{A_i}\right),\quad
\ket{c_1^i}=\tfrac{1}{\sqrt{2}}\left(\tfrac{1}{\sqrt{\braket{x_i|x_i}}}\ket{x_i}\ket{0}_{A_i}+\ket{\psiwn}\ket{1}_{A_i}\right),
\end{equation}
where $\ket{x_i}$ are the result of optimization \eqref{aaaa}. As $\ket{\psiwn},\ket{x_i}\in\mH_\mC$, obviously the noise acts trivially on $\mH_{\mC_i}$. Moreover:
\begin{equation}
\begin{split}
\bra{c_0^i}G_j\otimes\openone\ket{c_1^i}&=\tfrac{1}{2\sqrt{\braket{x_i|x_i}}}\left(\bra{\psiwn}G_j^\mC\ket{x_i}+\bra{x_i}G_j^\mC\ket{\psiwn}\right)=\tfrac{1}{T\sqrt{\braket{x_i|x_i}}}\delta_{ij}\\
\bra{c_0^i}G_j\otimes\openone\ket{c_0^i}&=\tfrac{1}{2}\left(\bra{\psiwn}G_j^\mC\ket{\psiwn}+\tfrac{1}{\braket{x_i|x_i}}\bra{x_i}G_j^\mC\ket{x_i}\right)=\bra{c_1^i}G_j\otimes\openone\ket{c_1^i},
\end{split}
\end{equation}
so also condition $\forall_{j\neq i}\Pi_{\mH_{\mC_i}}G_j\Pi_{\mH_{\mC_i}}\propto \Pi_{\mH_{\mC_i}}$ is satisfied.

Therefore SEP-QEC with initial state $\rho_{\rm in}=\frac{1}{P}\sum_{i=1}^P\ket{c_0^i}\bra{c_0^i}$ and $\mH_{\mC_i}$ defined above leads to the total cost $\cost_\SEP=\sum_{i=1}^P P\cdot \weight_{ii}\braket{x_i|x_i}=P\cdot \cost_\JNT$. From this reasoning we see clearly that the largest possible advantage of JNT-QEC over SEP-QEC is to decreasing the total cost by a factor $P$.

\section{Optimality of SEP-QEC in the single-qubit model}
\label{sec:qubit}
Let us consider the most general case of a two-parameter quantum estimation problem under Markovian noise in a $2$-dimensional Hilbert space when the HS achievability condition is satisfied. Without loss of generality we assume that the Lindblad operator is $L=\sigma_z$,  the Hamiltonian is $H=\omega_x\sigma_x+\omega_y\sigma_y$ and the cost matrix is diagonal:
\begin{equation}
\begin{bmatrix}
\weight_{xx} & 0\\
0 & \weight_{yy}
\end{bmatrix}
\end{equation}
(otherwise, one can always apply a proper transformation in the parameter space $\bomega'=\bomega A^{-1}$ which diagonalizes the cost matrix, without changing orthonormality of the generators). Below we show that for such a problem there is no advantage of JNT-QEC over SEP-QEC.

First let us consider the optimal SEP-QEC. Each generator has eigenvalues $\ket{\lambda_{x/y\pm}}=\pm1$ and each parameter may be estimated with precision  $\Delta^2\tilde\omega_{x/y}=\frac{1}{4T^2}$~\cite{sekatski2017quantum}. Moreover, as the cost is diagonal, there is no point in applying an additional transformation in the optimization procedure given in \eqref{optsep}---indeed, any matrix $A$ satisfying $\forall_i (AA^T)_{ii}=1$ preserves the eigenvalues of the generators as well as $\trace(A\weight A^T)$. Therefore the optimal SEP-QEC cost is simply given by $\cost_\SEP=\frac{(\sqrt{\weight_{xx}}+\sqrt{\weight_{yy}})^2}{4T^2}$.
We will show below that this  cannot be outperformed by the optimal JNT-QECs.

For the sake of notation simplicity, when tensoring an operator acting on $\mH_S$ with the identity on $\mH_A$, the part $\otimes\openone$ will be omitted and we will denote $\sigma_i\otimes\openone$ simply as $\sigma_i$ (unless this leads to ambiguity).

First, we note that the diagonal elements of the QFI matrix for the state $\ket{\psiw}$ are
\begin{equation}
\begin{split}
\label{spf}
\fisher_{ii}&=4T^2(\braket{\psiw|\sigma_i\Pi_{\mH_\mC}\sigma_i|\psiw}-\abs{\braket{\psiw|\Pi_{\mH_\mC}\sigma_i|\psiw}}^2)\\ &\leq 4T^2\braket{\psiw|\sigma_i\Pi_{\mH_\mC}\sigma_i|\psiw}\quad(i=x,y).
\end{split}
\end{equation}
Moreover
\begin{equation}
\begin{split}
\label{fineq}
\cost&=\weight_{xx}\Delta^2\tilde\omega_x+\weight_{yy}\Delta^2\tilde\omega_y\geq \frac{\weight_{xx}}{\fisher_{xx}}+\frac{\weight_{yy}}{\fisher_{yy}}=\frac{1}{\fisher_{xx}+\fisher_{yy}}\left(\frac{\weight_{xx}}{\frac{\fisher_{xx}}{\fisher_{xx}+\fisher_{yy}}}+\frac{\weight_{yy}}{\frac{\fisher_{yy}}{\fisher_{xx}+\fisher_{yy}}}\right)\\
&\geq\frac{1}{\fisher_{xx}+\fisher_{yy}}\min_{p\in[0,1]}\left(\frac{\weight_{xx}}{p}+\frac{\weight_{yy}}{1-p}\right)=\frac{(\sqrt{\weight_{xx}}+\sqrt{\weight_{yy}})^2}{\fisher_{xx}+\fisher_{yy}}.
\end{split}
\end{equation}
Therefore we may focus on an upper bound for $\sum_{i=x,y}\braket{\psiw|\sigma_i\Pi_{\mH_\mC}\sigma_i|\psiw}$.

Let $\{\ket{c_0},\ket{c_1},\ket{c_2}\}$ be an orthonormal basis of $\mH_\mC\subseteq\mH_S\otimes\mH_A$. These vectors can be written down as
\begin{equation}
\ket{c_i}=\cos(\varphi^i)\ket{0}\ket{A^i_0}+\sin(\varphi^i)\ket{1}\ket{A^i_1},
\end{equation}
where $\ket{A^i_{0/1}}$ are normalized states in $\mH_A$ and $\varphi^i\in[0,\frac{\pi}{2}]$ (a potential complex phase is incorporated in the definition of$\ket{A^i_{0/1}}$). The QEC condition requires $\forall_{i,j}\braket{c_i|\sigma_z|c_j}=\lambda\delta_{ij}$, which leads to the following two constraints: (i) $\forall_i\cos^2(\varphi^i)-\sin^2(\varphi^i)=\lambda$ means all $\varphi^i$ are equal (therefore superscript $i$ will be omitted); (ii) $\forall_{i\neq j}\cos^2(\varphi)\braket{A^i_0|A^j_0}-\sin^2(\varphi)\braket{A^i_1|A^j_1}=0$. Together with the orthonormality condition $\forall_{i\neq j}\cos^2(\varphi)\braket{A^i_0|A^j_0}+\sin^2(\varphi)\braket{A^i_1|A^j_1}=0$, we have
\begin{equation}
\begin{split}
\label{eq:Aortho}
&\ket{c_i}=\cos(\varphi)\ket{0}\ket{A^i_0}+\sin(\varphi)\ket{1}\ket{A^i_1},\\
&\quad\forall_{i,j}\braket{A_0^i|A_0^j}=\delta_{ij},\quad \braket{A_1^i|A_1^j}=\delta_{ij}.
\end{split}
\end{equation}
Note that there is no fixed relationship between sets $\{\ket{A_0^i}\}_{i=0}^{i=2}$ and $\{\ket{A_1^i}\}_{i=0}^{i=2}$---in particular it may happen that  $\vspan\{\ket{A_0^i}\}\neq\vspan\{\ket{A_1^i}\}$. Effective generators in the chosen basis are given as:
\begin{align}
(G^\eff_x)_{ji} & =  \braket{c_j|\sigma_x|c_i} = \frac{\sin(2\varphi)}{2}(\braket{A_0^j|A_1^i}+\braket{A_1^j|A_0^i}),\\
(G^\eff_y)_{ji} & =  \braket{c_j|\sigma_y|c_i} = i\frac{\sin(2\varphi)}{2}(\braket{A_1^j|A_0^i}-\braket{A_0^j|A_1^i}).
\end{align}
We focus on estimation around point $\bomega=[0,...,0]$ for which $\ket{\psiw}=\ket{\psi_{\rm in}}=\ket{c_0}$.
Then
\begin{multline}
\braket{\psiw|\sigma_{{x}\slash{y}}\Pi_{\mH_\mC}\sigma_{{x}\slash{y}}|\psiw}=\sum_{i=0}^2\braket{c_0|\sigma_{{x}\slash{y}}|c_i}\braket{c_i|\sigma_{{x}\slash{y}}|c_0}=\\
\frac{\sin^2(2\varphi)}{4}
\big(\sum_{i=0}^2 \abs{\braket{A_0^0|A_1^i}}^2+\abs{\braket{A_1^0|A_0^i}}^2\pm2\real(\braket{A_1^0|A_0^i}\braket{A_1^i|A_0^0})\big).
\end{multline}
Since for each $k=0/1$, states $\{\ket{A_{k}^i}\}_{i=0}^2$ are mutually orthonormal,
\begin{equation}
\sum_{i=x,y}\braket{\psi|\sigma_i\Pi_{\mH_\mC}\sigma_i|\psi} =
\frac{\sin^2(2\varphi)}{2} \cdot \sum_{i=0}^2(\abs{\braket{A_0^0|A_1^i}}^2+\abs{\braket{A_1^0|A_0^i}}^2)\leq \sin^2(2\varphi) \leq 1,
\end{equation}
where the first inequality is saturated if and only if both $\ket{A^0_0}\in\vspan\{\ket{A^i_1}\}_{i=0}^2$ and $\ket{A^0_1}\in\vspan\{\ket{A^i_0}\}_{i=0}^2$. Using \eqref{spf} and \eqref{fineq} we get $\cost \geq \frac{(\sqrt{\weight_{xx}}+\sqrt{\weight_{yy}})^2}{4T^2}$. This implies that the JNT-QEC stategy cannot outperform the best SEP-QEC strategy.

\section{The optimal SEP protocol for sensing all magnetic field components in presence of correlated dephasing noise in the two-qubit model}
\label{sec:2qubits}
Here we prove formally, that the optimal precision achievable in the second example from the main text is $\cost_{\rm SEP}=\frac{9}{4T^2}$.

Let us briefly review the problem. We consider a two-atom system with Hamiltonian $H=\frac{1}{2}\sum_{i=1}^2\bomega \cdot\bsigma^{(i)}$ (where $\bsigma^{(i)}=[\sigma_x^{(i)},\sigma_y^{(i)},\sigma_z^{(i)}]$ acts on the $i^{\text{th}}$ atom) and a single Lindblad operator $L=\sqrt{2\gamma}(\sigma_z^{(1)}-\sigma_z^{(2)})$, therefore
\begin{equation}
\label{tqs}
\mS=\vspan_{\mathbb R}\{\openone,L,L^2\}=\vspan_{\mathbb R}\{\openone,\sigma_z^{(1)}-\sigma_z^{(2)},\sigma_z^{(1)}\sigma_z^{(2)}\}.
\end{equation}
As for all three generators the minimal and the maximal eigenvalues are respectively $-1$ and $+1$, we immediately see that $\Delta^2\tilde\omega_{x,y,z}\geq \frac{1}{4T^2}$ and from that $\cost_{\rm SEP}\geq\frac{P^2}{4T^2}=\frac{9}{4T^2}$ (which cannot be improved by applying transformation $A$). Below we show that such a precision is indeed achievable.

First, $\omega_z$ can be estimated using a decoherence-free subspace~\cite{lidar1998decoherence} $\vspan\{\ket{\psi_z},\frac{1}{2}(\sigma_z^{(1)}+\sigma_z^{(2)})\ket{\psi_z}\}$, where $\ket{\psi_z}=\frac{1}{\sqrt{2}}(\ket{\up\up}+\ket{\down\down})$, which leads to the precision $\Delta^2\omega_{z}=\frac{1}{4T^2}$.

In case of $\omega_x$ the situation is slightly more complicated, as using the analogue approach the subspace $\vspan\{\ket{\psi_x},\frac{1}{2}(\sigma_x^{(1)}+\sigma_x^{(2)})\ket{\psi_x}\}$ would not satisfy the QEC conditions. To get the desired estimation precision we need to find the state $\ket{\psi_x}$ which is optimal for measuring $\omega_x$ (from the point of view of noiseless, single-parameter estimation) and for which
\begin{equation}
\label{csx}
\mH_{\mC_x}=\vspan\{\ket{\psi_x},\frac{1}{2}(\sigma_x^{(1)}+\sigma_x^{(2)})\ket{\psi_x}\}
\end{equation}
satisfies the QEC condition:
\begin{equation}
\begin{split}
\label{eqec}
\Pi_{\mH_{\mC_x}}(\sigma_z^{(1)}-\sigma_z^{(2)})\Pi_{\mH_{\mC_x}}\propto\Pi_{\mH_{\mC_x}},\\
\Pi_{\mH_{\mC_x}}\sigma_z^{(1)}\sigma_z^{(2)}\Pi_{\mH_{\mC_x}}\propto\Pi_{\mH_{\mC_x}},
\end{split}
\end{equation}
and, moreover, other generators act trivially inside $\mH_{\mC_x}$:
\begin{equation}
\begin{split}
\label{esep}
\Pi_{\mH_{\mC_y}}(\sigma_y^{(1)}+\sigma_y^{(2)})\Pi_{\mH_{\mC_x}}\propto\Pi_{\mH_{\mC_x}},\\
\Pi_{\mH_{\mC_z}}(\sigma_z^{(1)}+\sigma_z^{(2)})\Pi_{\mH_{\mC_x}}\propto\Pi_{\mH_{\mC_x}}.
\end{split}
\end{equation}
It is known that for single-parameter frequency estimation the optimal state corresponds to an equally weighted superposition of states with minimal and maximal eigenvalues of the generator. For $\frac{1}{2}(\sigma_x^{(1)}+\sigma_x^{(2)})$ it will be:
\begin{equation}
\ket{\psi_x^{\varphi}}=\frac{1}{2\sqrt{2}}\cdot\left((\ket{\up}+\ket{\down})(\ket{\up}+\ket{\down})+e^{i\varphi}(\ket{\up}-\ket{\down})(\ket{\up}-\ket{\down})\right)
\end{equation}
for any $\varphi\in\mathbb R$. Note, that any superposition of $\ket{\psi_x^{\varphi}}$ (with different $\varphi$) entangled with separated ancillae is still optimal for sensing $\omega_x$. Therefore, we can take $\ket{\psi_x}=\frac{1}{\sqrt{2}}(\ket{\psi_x^0}\ket{1}_A+\ket{\psi_x^{\pi}}\ket{2}_A)$. We have:
\begin{equation}
\ket{\psi_x}=\frac{1}{2}\big((\ket{\up\up}+\ket{\down\down})\ket{1}_A+(\ket{\up\down}+\ket{\down\up})\ket{2}_A\big),
\end{equation}
\begin{equation}
\frac{1}{2}(\sigma_x^{(1)}+\sigma_x^{(2)})\ket{\psi_x}=\frac{1}{2}\big((\ket{\up\down}+\ket{\down\up})\ket{1}_A+(\ket{\up\up}+\ket{\down\down})\ket{2}_A\big).
\end{equation}
and then the code space \eqref{csx} satisfies \eqref{eqec} and \eqref{esep}. It gives as $\Delta^2\omega_{x }=\frac{1}{4T^2}$. Analogous reasoning could be provided for $\omega_y$.

Therefore for $\rho_{\rm in} = \frac{1}{3}\sum_{i=x,y,z}\ket{\psi_i}\bra{\psi_i}$ we have $\cost=\frac{9}{4T^2}$  in line with general considerations on the performance of the SEP-QEC codes as given in Sec.~\ref{sec:JNTSEP}.

\section{Estimating the $SU(d)$ generators}
\label{sec:sud}


Below we present an example of a JNT-QEC protocol allowing one to achieve the total cost $\cost = \Theta(P^{\frac{3}{2}})$ for the last example in the main text. For clarification, we treat $d$-dimensional Hilbert space as a single spin-$j$ particle ($d=2j+1$) and we use the notation where $\{\ket{k}\}_{k=-j}^j$ is the  eigenbasis of the $J_z$ operator.

We consider a problem where the noise generator $J_z$ and the unitary evolution $H$ read:
\begin{equation}
J_z=\sum_{k=-j}^j k\ket{k}\bra{k},\qquad H=\sum_{i=1}^P\omega_{i}G_{i},
\end{equation}
where $G_i$ is an orthonormal basis of $\mS^\perp$---the orthogonal complement of $\mS = \vspan\{\id,J_z,J_z^2\}$ (therefore $P=d^2-3$). For technical reasons we distinguish three groups of operators that form the basis $\{G_i\}$:
\begin{itemize}
\item Real off-diagonal: $G^R_{kl}=\frac{1}{\sqrt{2}}(\ket{k}\bra{l}+\ket{l}\bra{k})$
\item Imaginary off-diagonal: $G^I_{kl}=\frac{i}{\sqrt{2}}(\ket{k}\bra{l}-\ket{l}\bra{k})$
\item Diagonal: $G^D_{i}=\sum_{k=-j}^jg_i^k\ket{k}\bra{k}$
\end{itemize}
and in what follows we prove the scaling $\Theta(P^{\frac{3}{2}})$ for each group. For simplicity, we assume, that $j$ is an integer (for half-integer $j$ the proof remains almost the same) and in this section we focus on the estimation around point $\bomega=[0,\ldots,0]$ and set $T=1$.

\noindent\textbf{Real off-diagonal generators.} We take $\dim(\mH_A)=\dim(\mH_S)$ and the state $\ket{\psi_\bomega}=\ket{\psi^R}=\frac{1}{\sqrt{2j+1}}\sum_{k=-j}^j\ket{k}\ket{k}_A\in\mH_S\otimes\mH_{A}$, we have
\begin{equation}
\braket{\psi^R|J_z|\psi^R}=0,\quad \braket{\psi^R|J_z^2|\psi^R}
=\frac{j(j+1)}{3}.
\end{equation}
We construct the code space in the following way. First, we act on $\ket{\psi}$ with generators: $G_{kl}^R\ket{\psi}=\frac{1}{\sqrt{2}}(\ket{k}\ket{l}_A+\ket{l}\ket{k}_A)$ and then we ``fix it'' to satisfy the QEC condition by extending ancilla $\mH_A\rightarrow\mH_A\oplus\mH_B$ and adding more terms:
\begin{equation}
\ket{c^R_{kl}}=\frac{p}{\sqrt{2}}(\ket{k}\ket{l}_A+\ket{l}\ket{k}_A)+
+q\ket{j}\ket{klj}_B+r\ket{-j}\ket{kl(-j)}_B+s\ket{0}\ket{kl0}_B,
\end{equation}
where $\braket{klm|k'l'm'}_B=\delta_{(klm)(k'l'm')}$. Then the QEC condition is equivalent to:
\begin{equation}
\label{eq:con}
\begin{split}
&\braket{c^R_{kl}|c^R_{kl}}=p^2+q^2+r^2+s^2=1,\\
&\braket{c^R_{kl}|J_z|c^R_{kl}}=\frac{p^2}{2}(k+l)+(q^2-r^2)j=0,\\
&\braket{c^R_{kl}|J_z^2|c^R_{kl}}=\frac{p^2}{2}(k^2+l^2)+(q^2+r^2)j^2=\frac{j(j+1)}{3}.\\
\end{split}
\end{equation}
The off-diagonal terms are automatically zero, no matter what $p,q,r,s$ are. We can write down $q^2$, $r^2$ and $s^2$ as linear functions of $p^2$:
\begin{equation}
\begin{split}
q^2&=\frac{1}{2j^2}\Big(\frac{j(j+1)}{3}-\frac{p^2}{2}(k^2+l^2+j(k+l))\Big),\\
r^2&=\frac{1}{2j^2}\Big(\frac{j(j+1)}{3}-\frac{p^2}{2}(k^2+l^2-j(k+l))\Big),\\
s^2&=1-p^2-\frac{1}{j^2}\Big(\frac{j(j+1)}{3}-\frac{p^2}{2}(k^2+l^2)\Big).
\end{split}
\end{equation}
Note that $p$ is a valid coefficient if the above set of equations has a solution (i.e. if the right-hand sides are positive). As $-2j\leq k+l\leq 2j$, $k^2+l^2\leq 2j^2$, this always holds provided $p^2=\frac{1}{6}$. For the code space ${\mH_{\mC_R}}$ spanned by vectors constructed in such a way, we have
\begin{equation}
(G_{kl}^{R})^{{\mH_{\mC_R}}}\ket{\psi^R}=\frac{p}{\sqrt{2j+1}}\ket{c^R_{kl}}.
\end{equation}
The QFIs are
\begin{equation}
\fisher^R_{(kl)(k'l')}=4\real(\braket{\psi^R|(G_{kl}^{R})^{{\mH_{\mC_R}}}(G_{k'l'}^{R})^{{\mH_{\mC_R}}}|\psi^R}-\braket{\psi^R|(G_{kl}^{R})^{{\mH_{\mC_R}}}|\psi^R}\braket{\psi^R|(G_{k'l'}^{R})^{{\mH_{\mC_R}}}|\psi^R})
\end{equation}
 which in our case simplifies to:
\begin{equation}
\fisher^R_{(kl)(k'l')}=\delta_{(kl)(k'l')}4\braket{\psi^R|((G_{kl}^{R})^{{\mH_{\mC_R}}})^2|\psi^R}=\frac{4p^2}{2j+1}.
\end{equation}
As $\braket{\psi^R|[(G_{kl}^{R})^{{\mH_{\mC_R}}},(G_{k'l'}^{R})^{{\mH_{\mC_R}}}]|\psi^R}=0$, the CR bound is saturable and the total cost is
\begin{equation}
\sum_{k>l}\Delta^2\omega^R_{kl}=\frac{2j+1}{4}\sum_{k>l}\frac{1}{p^2}=\frac{3j(2j+1)^2}{2}=\Theta(P^{\frac{3}{2}}).
\end{equation}

\noindent\textbf{Imaginary off-diagonal generators.} The reasoning is analogous to the previous case. Note that using different ancillary spaces for real and imaginary generators is needed. Even though $\braket{\psi^R|G^R_{kl}J_zG^R_{k'l'}|\psi^R}\propto \delta_{kl,k'l'}$ and $\braket{\psi^R|G^I_{kl}J_zG^I_{k'l'}|\psi^R}\propto \delta_{kl,k'l'}$ are satisfied automatically,\\ $\braket{\psi^R|G^R_{kl}J_zG^I_{k'l'}|\psi^R}\propto \delta_{kl,k'l'}$ may not be true. 

\noindent\textbf{Diagonal generators.} The number of diagonal generators scales like $\Theta(j)$ (whereas for off-diagonal elements the scaling is $\Theta(j^2)$), implying that the estimation with respect to diagonal generators does not contribute significantly to the overall scaling. Therefore we could simply use the SEP-QEC approach. Following \cite{zhou2018achieving}, any traceless generator may by written down as:
\begin{equation}
G^D_i=\frac{1}{2}\trace(|G^D_i|)(\rho_{i+}-\rho_{i-}).
\end{equation}
We define states $\ket{c_{i+}},\ket{c_{i-}}$ as purifications of these density matrices by using mutually orthogonal ancillary subspaces $\mH_{Ai+},\mH_{Ai-}$:
\begin{equation}
\rho_{i+/-}=\trace_{Ai+/-}(\ket{c_{i+/-}}\bra{c_{i+/-}}).
\end{equation}
Therefore
\begin{equation}
\begin{split}
\bra{c_{i+}}G^D_i\ket{c_{i+}}&=\frac{1}{2}\trace(|G^D_i|)\geq\frac{1}{2\sqrt{2j+1}},\\
\bra{c_{i-}}G^D_i\ket{c_{i-}}&=-\frac{1}{2}\trace(|G^D_i|)\leq-\frac{1}{2\sqrt{2j+1}},
\end{split}
\end{equation}
and from that for code space $\vspan\{\ket{c_{i+}},\ket{c_{i-}}\}$ and input state $\frac{1}{\sqrt{2}}(\ket{c_{i+}}+\ket{c_{i-}})$ we have $F_{\omega_i}\geq \frac{1}{2j+1}$. For a single-parameter problem the CR bound is always saturable
so using the SEP-QEC approach for the input state
\begin{equation}
\rho_{\rm in}^D=\frac{1}{2j-2}\sum_{i=1}^{2j-2}\frac{1}{2}\left(\ket{c_{i+}}+\ket{c_{i-}}\right)\left(\bra{c_{i+}}+\bra{c_{i-}}\right)
\end{equation}
we have
\begin{equation}
\sum_{i=1}^{2j-2}\Delta^2\omega^D_i=(2j-2)\sum_{i=1}^{2j-2}\frac{1}{F_{\omega_i}}\leq (2j+1)(2j-2)^2 = \Theta(P^{\frac{3}{2}}).
\end{equation}

\noindent\textbf{Results.}
Finally, combining all above, for the input state:
\begin{equation}
\rho_{\rm in}=\frac{1}{3}(\ket{\psi^R}\bra{\psi^R}+\ket{\psi^I}\bra{\psi^I}+\rho^D)
\end{equation}
(with properly applied QEC protocol) we get
\begin{equation}
\sum_{i=1}^{P}\Delta^2\omega_i=\Theta(P^{\frac{3}{2}}).
\end{equation}

\bibliographystyle{aps}

\begin{thebibliography}{76}%
\makeatletter
\providecommand \@ifxundefined [1]{%
 \@ifx{#1\undefined}
}%
\providecommand \@ifnum [1]{%
 \ifnum #1\expandafter \@firstoftwo
 \else \expandafter \@secondoftwo
 \fi
}%
\providecommand \@ifx [1]{%
 \ifx #1\expandafter \@firstoftwo
 \else \expandafter \@secondoftwo
 \fi
}%
\providecommand \natexlab [1]{#1}%
\providecommand \enquote  [1]{``#1''}%
\providecommand \bibnamefont  [1]{#1}%
\providecommand \bibfnamefont [1]{#1}%
\providecommand \citenamefont [1]{#1}%
\providecommand \href@noop [0]{\@secondoftwo}%
\providecommand \href [0]{\begingroup \@sanitize@url \@href}%
\providecommand \@href[1]{\@@startlink{#1}\@@href}%
\providecommand \@@href[1]{\endgroup#1\@@endlink}%
\providecommand \@sanitize@url [0]{\catcode `\\12\catcode `\$12\catcode
  `\&12\catcode `\#12\catcode `\^12\catcode `\_12\catcode `\%12\relax}%
\providecommand \@@startlink[1]{}%
\providecommand \@@endlink[0]{}%
\providecommand \url  [0]{\begingroup\@sanitize@url \@url }%
\providecommand \@url [1]{\endgroup\@href {#1}{\urlprefix }}%
\providecommand \urlprefix  [0]{URL }%
\providecommand \Eprint [0]{\href }%
\providecommand \doibase [0]{http://dx.doi.org/}%
\providecommand \selectlanguage [0]{\@gobble}%
\providecommand \bibinfo  [0]{\@secondoftwo}%
\providecommand \bibfield  [0]{\@secondoftwo}%
\providecommand \translation [1]{[#1]}%
\providecommand \BibitemOpen [0]{}%
\providecommand \bibitemStop [0]{}%
\providecommand \bibitemNoStop [0]{.\EOS\space}%
\providecommand \EOS [0]{\spacefactor3000\relax}%
\providecommand \BibitemShut  [1]{\csname bibitem#1\endcsname}%
\let\auto@bib@innerbib\@empty
\bibitem [{\citenamefont {Giovannetti}\ \emph {et~al.}(2006)\citenamefont
  {Giovannetti}, \citenamefont {Lloyd},\ and\ \citenamefont
  {Maccone}}]{giovannetti2006quantum}%
  \BibitemOpen
  \bibfield  {author} {\bibinfo {author} {\bibfnamefont {V.}~\bibnamefont
  {Giovannetti}}, \bibinfo {author} {\bibfnamefont {S.}~\bibnamefont {Lloyd}},
  \ and\ \bibinfo {author} {\bibfnamefont {L.}~\bibnamefont {Maccone}},\
  }\bibfield  {title} {\emph {\bibinfo {title} {Quantum metrology},\ }}\href
  {https://doi.org/10.1103/PhysRevLett.96.010401} {\bibfield  {journal}
  {\bibinfo  {journal} {Phys. Rev. Lett.}\ }\textbf {\bibinfo {volume} {96}},\
  \bibinfo {pages} {010401} (\bibinfo {year} {2006})}\BibitemShut {NoStop}%
\bibitem [{\citenamefont {Paris}(2009)}]{Paris2009}%
  \BibitemOpen
  \bibfield  {author} {\bibinfo {author} {\bibfnamefont {M.~G.~A.}\
  \bibnamefont {Paris}},\ }\bibfield  {title} {\emph {\bibinfo {title} {Quantum
  estimation for quantum technologies},\ }}\href
  {https://doi.org/10.1142/S0219749909004839} {\bibfield  {journal} {\bibinfo
  {journal} {Int. J. Quantum Inf.}\ }\textbf {\bibinfo {volume} {07}},\
  \bibinfo {pages} {125} (\bibinfo {year} {2009})}\BibitemShut {NoStop}%
\bibitem [{\citenamefont {Giovannetti}\ \emph {et~al.}(2011)\citenamefont
  {Giovannetti}, \citenamefont {Lloyd},\ and\ \citenamefont
  {Maccone}}]{giovannetti2011advances}%
  \BibitemOpen
  \bibfield  {author} {\bibinfo {author} {\bibfnamefont {V.}~\bibnamefont
  {Giovannetti}}, \bibinfo {author} {\bibfnamefont {S.}~\bibnamefont {Lloyd}},
  \ and\ \bibinfo {author} {\bibfnamefont {L.}~\bibnamefont {Maccone}},\
  }\bibfield  {title} {\emph {\bibinfo {title} {Advances in quantum
  metrology},\ }}\href {https://doi.org/10.1038/nphoton.2011.35} {\bibfield
  {journal} {\bibinfo  {journal} {Nat. Photonics}\ }\textbf {\bibinfo {volume}
  {5}},\ \bibinfo {pages} {222} (\bibinfo {year} {2011})}\BibitemShut {NoStop}%
\bibitem [{\citenamefont {Toth}\ and\ \citenamefont
  {Apellaniz}(2014)}]{Toth2014}%
  \BibitemOpen
  \bibfield  {author} {\bibinfo {author} {\bibfnamefont {G.}~\bibnamefont
  {Toth}}\ and\ \bibinfo {author} {\bibfnamefont {I.}~\bibnamefont
  {Apellaniz}},\ }\bibfield  {title} {\emph {\bibinfo {title} {Quantum
  metrology from a quantum information science perspective},\ }}\href
  {http://doi.org/10.1088/1751-8113/47/42/424006} {\bibfield  {journal}
  {\bibinfo  {journal} {J. Phys. A: Math. Theor.}\ }\textbf {\bibinfo {volume}
  {47}},\ \bibinfo {pages} {424006} (\bibinfo {year} {2014})}\BibitemShut
  {NoStop}%
\bibitem [{\citenamefont {Demkowicz-Dobrza{\'n}ski}\ \emph
  {et~al.}(2015)\citenamefont {Demkowicz-Dobrza{\'n}ski}, \citenamefont
  {Jarzyna},\ and\ \citenamefont {Ko\l{}ody\'{n}ski}}]{Demkowicz2015}%
  \BibitemOpen
  \bibfield  {author} {\bibinfo {author} {\bibfnamefont {R.}~\bibnamefont
  {Demkowicz-Dobrza{\'n}ski}}, \bibinfo {author} {\bibfnamefont
  {M.}~\bibnamefont {Jarzyna}}, \ and\ \bibinfo {author} {\bibfnamefont
  {J.}~\bibnamefont {Ko\l{}ody\'{n}ski}},\ }in\ \href
  {https://doi.org/10.1016/bs.po.2015.02.003} {\emph {\bibinfo {booktitle}
  {Prog. Optics}}},\ Vol.~\bibinfo {volume} {60},\ \bibinfo {editor} {edited
  by\ \bibinfo {editor} {\bibfnamefont {E.}~\bibnamefont {Wolf}}}\ (\bibinfo
  {publisher} {Elsevier},\ \bibinfo {year} {2015})\ pp.\ \bibinfo {pages}
  {345--435}\BibitemShut {NoStop}%
\bibitem [{\citenamefont {Schnabel}(2017)}]{Schnabel2016}%
  \BibitemOpen
  \bibfield  {author} {\bibinfo {author} {\bibfnamefont {R.}~\bibnamefont
  {Schnabel}},\ }\bibfield  {title} {\emph {\bibinfo {title} {Squeezed states
  of light and their applications in laser interferometers},\ }}\href
  {https://doi.org/10.1016/j.physrep.2017.04.001} {\bibfield  {journal}
  {\bibinfo  {journal} {Phys. Rep.}\ }\textbf {\bibinfo {volume} {684}},\
  \bibinfo {pages} {1 } (\bibinfo {year} {2017})}\BibitemShut {NoStop}%
\bibitem [{\citenamefont {Degen}\ \emph {et~al.}(2017)\citenamefont {Degen},
  \citenamefont {Reinhard},\ and\ \citenamefont
  {Cappellaro}}]{degen2017quantum}%
  \BibitemOpen
  \bibfield  {author} {\bibinfo {author} {\bibfnamefont {C.~L.}\ \bibnamefont
  {Degen}}, \bibinfo {author} {\bibfnamefont {F.}~\bibnamefont {Reinhard}}, \
  and\ \bibinfo {author} {\bibfnamefont {P.}~\bibnamefont {Cappellaro}},\
  }\bibfield  {title} {\emph {\bibinfo {title} {Quantum sensing},\ }}\href
  {https://doi.org/10.1103/RevModPhys.89.035002} {\bibfield  {journal}
  {\bibinfo  {journal} {Rev. Mod. Phys.}\ }\textbf {\bibinfo {volume} {89}},\
  \bibinfo {pages} {035002} (\bibinfo {year} {2017})}\BibitemShut {NoStop}%
\bibitem [{\citenamefont {Pezz\`e}\ \emph {et~al.}(2018)\citenamefont
  {Pezz\`e}, \citenamefont {Smerzi}, \citenamefont {Oberthaler}, \citenamefont
  {Schmied},\ and\ \citenamefont {Treutlein}}]{Pezze2018}%
  \BibitemOpen
  \bibfield  {author} {\bibinfo {author} {\bibfnamefont {L.}~\bibnamefont
  {Pezz\`e}}, \bibinfo {author} {\bibfnamefont {A.}~\bibnamefont {Smerzi}},
  \bibinfo {author} {\bibfnamefont {M.~K.}\ \bibnamefont {Oberthaler}},
  \bibinfo {author} {\bibfnamefont {R.}~\bibnamefont {Schmied}}, \ and\
  \bibinfo {author} {\bibfnamefont {P.}~\bibnamefont {Treutlein}},\ }\bibfield
  {title} {\emph {\bibinfo {title} {Quantum metrology with nonclassical states
  of atomic ensembles},\ }}\href {https://doi.org/10.1103/RevModPhys.90.035005}
  {\bibfield  {journal} {\bibinfo  {journal} {Rev. Mod. Phys.}\ }\textbf
  {\bibinfo {volume} {90}},\ \bibinfo {pages} {035005} (\bibinfo {year}
  {2018})}\BibitemShut {NoStop}%
\bibitem [{\citenamefont {Pirandola}\ \emph {et~al.}(2018)\citenamefont
  {Pirandola}, \citenamefont {Bardhan}, \citenamefont {Gehring}, \citenamefont
  {Weedbrook},\ and\ \citenamefont {Lloyd}}]{Pirandola2018}%
  \BibitemOpen
  \bibfield  {author} {\bibinfo {author} {\bibfnamefont {S.}~\bibnamefont
  {Pirandola}}, \bibinfo {author} {\bibfnamefont {B.~R.}\ \bibnamefont
  {Bardhan}}, \bibinfo {author} {\bibfnamefont {T.}~\bibnamefont {Gehring}},
  \bibinfo {author} {\bibfnamefont {C.}~\bibnamefont {Weedbrook}}, \ and\
  \bibinfo {author} {\bibfnamefont {S.}~\bibnamefont {Lloyd}},\ }\bibfield
  {title} {\emph {\bibinfo {title} {Advances in photonic quantum sensing},\
  }}\href {https://doi.org/10.1038/s41566-018-0301-6} {\bibfield  {journal}
  {\bibinfo  {journal} {Nat. Photonics}\ }\textbf {\bibinfo {volume} {12}},\
  \bibinfo {pages} {724} (\bibinfo {year} {2018})}\BibitemShut {NoStop}%
\bibitem [{\citenamefont {Caves}(1981)}]{caves1981quantum}%
  \BibitemOpen
  \bibfield  {author} {\bibinfo {author} {\bibfnamefont {C.~M.}\ \bibnamefont
  {Caves}},\ }\bibfield  {title} {\emph {\bibinfo {title} {Quantum-mechanical
  noise in an interferometer},\ }}\href
  {https://doi.org/10.1103/PhysRevD.23.1693} {\bibfield  {journal} {\bibinfo
  {journal} {Phys. Rev. D}\ }\textbf {\bibinfo {volume} {23}},\ \bibinfo
  {pages} {1693} (\bibinfo {year} {1981})}\BibitemShut {NoStop}%
\bibitem [{\citenamefont {Holland}\ and\ \citenamefont
  {Burnett}(1993)}]{holland1993interferometric}%
  \BibitemOpen
  \bibfield  {author} {\bibinfo {author} {\bibfnamefont {M.}~\bibnamefont
  {Holland}}\ and\ \bibinfo {author} {\bibfnamefont {K.}~\bibnamefont
  {Burnett}},\ }\bibfield  {title} {\emph {\bibinfo {title} {Interferometric
  detection of optical phase shifts at the heisenberg limit},\ }}\href
  {https://doi.org/10.1103/PhysRevLett.71.1355} {\bibfield  {journal} {\bibinfo
   {journal} {Phys. Rev. Lett.}\ }\textbf {\bibinfo {volume} {71}},\ \bibinfo
  {pages} {1355} (\bibinfo {year} {1993})}\BibitemShut {NoStop}%
\bibitem [{\citenamefont {Lee}\ \emph {et~al.}(2002)\citenamefont {Lee},
  \citenamefont {Kok},\ and\ \citenamefont {Dowling}}]{lee2002quantum}%
  \BibitemOpen
  \bibfield  {author} {\bibinfo {author} {\bibfnamefont {H.}~\bibnamefont
  {Lee}}, \bibinfo {author} {\bibfnamefont {P.}~\bibnamefont {Kok}}, \ and\
  \bibinfo {author} {\bibfnamefont {J.~P.}\ \bibnamefont {Dowling}},\
  }\bibfield  {title} {\emph {\bibinfo {title} {A quantum rosetta stone for
  interferometry},\ }}\href {https://doi.org/10.1080/0950034021000011536}
  {\bibfield  {journal} {\bibinfo  {journal} {J. Mod. Optic.}\ }\textbf
  {\bibinfo {volume} {49}},\ \bibinfo {pages} {2325} (\bibinfo {year}
  {2002})}\BibitemShut {NoStop}%
\bibitem [{\citenamefont {Wineland}\ \emph {et~al.}(1992)\citenamefont
  {Wineland}, \citenamefont {Bollinger}, \citenamefont {Itano}, \citenamefont
  {Moore},\ and\ \citenamefont {Heinzen}}]{wineland1992spin}%
  \BibitemOpen
  \bibfield  {author} {\bibinfo {author} {\bibfnamefont {D.}~\bibnamefont
  {Wineland}}, \bibinfo {author} {\bibfnamefont {J.}~\bibnamefont {Bollinger}},
  \bibinfo {author} {\bibfnamefont {W.}~\bibnamefont {Itano}}, \bibinfo
  {author} {\bibfnamefont {F.}~\bibnamefont {Moore}}, \ and\ \bibinfo {author}
  {\bibfnamefont {D.}~\bibnamefont {Heinzen}},\ }\bibfield  {title} {\emph
  {\bibinfo {title} {Spin squeezing and reduced quantum noise in
  spectroscopy},\ }}\href {http://doi.org/10.1103/PhysRevA.46.R6797} {\bibfield
   {journal} {\bibinfo  {journal} {Phys. Rev. A}\ }\textbf {\bibinfo {volume}
  {46}},\ \bibinfo {pages} {R6797} (\bibinfo {year} {1992})}\BibitemShut
  {NoStop}%
\bibitem [{\citenamefont {McKenzie}\ \emph {et~al.}(2002)\citenamefont
  {McKenzie}, \citenamefont {Shaddock}, \citenamefont {McClelland},
  \citenamefont {Buchler},\ and\ \citenamefont
  {Lam}}]{mckenzie2002experimental}%
  \BibitemOpen
  \bibfield  {author} {\bibinfo {author} {\bibfnamefont {K.}~\bibnamefont
  {McKenzie}}, \bibinfo {author} {\bibfnamefont {D.~A.}\ \bibnamefont
  {Shaddock}}, \bibinfo {author} {\bibfnamefont {D.~E.}\ \bibnamefont
  {McClelland}}, \bibinfo {author} {\bibfnamefont {B.~C.}\ \bibnamefont
  {Buchler}}, \ and\ \bibinfo {author} {\bibfnamefont {P.~K.}\ \bibnamefont
  {Lam}},\ }\bibfield  {title} {\emph {\bibinfo {title} {Experimental
  demonstration of a squeezing-enhanced power-recycled michelson interferometer
  for gravitational wave detection},\ }}\href
  {https://doi.org/10.1103/PhysRevLett.88.231102} {\bibfield  {journal}
  {\bibinfo  {journal} {Phys. Rev. Lett.}\ }\textbf {\bibinfo {volume} {88}},\
  \bibinfo {pages} {231102} (\bibinfo {year} {2002})}\BibitemShut {NoStop}%
\bibitem [{\citenamefont {Bollinger}\ \emph {et~al.}(1996)\citenamefont
  {Bollinger}, \citenamefont {Itano}, \citenamefont {Wineland},\ and\
  \citenamefont {Heinzen}}]{bollinger1996optimal}%
  \BibitemOpen
  \bibfield  {author} {\bibinfo {author} {\bibfnamefont {J.}~\bibnamefont
  {Bollinger}}, \bibinfo {author} {\bibfnamefont {W.~M.}\ \bibnamefont
  {Itano}}, \bibinfo {author} {\bibfnamefont {D.}~\bibnamefont {Wineland}}, \
  and\ \bibinfo {author} {\bibfnamefont {D.}~\bibnamefont {Heinzen}},\
  }\bibfield  {title} {\emph {\bibinfo {title} {Optimal frequency measurements
  with maximally correlated states},\ }}\href
  {https://doi.org/10.1103/PhysRevA.54.R4649} {\bibfield  {journal} {\bibinfo
  {journal} {Phys. Rev. A}\ }\textbf {\bibinfo {volume} {54}},\ \bibinfo
  {pages} {R4649} (\bibinfo {year} {1996})}\BibitemShut {NoStop}%
\bibitem [{\citenamefont {Leibfried}\ \emph {et~al.}(2004)\citenamefont
  {Leibfried}, \citenamefont {Barrett}, \citenamefont {Schaetz}, \citenamefont
  {Britton}, \citenamefont {Chiaverini}, \citenamefont {Itano}, \citenamefont
  {Jost}, \citenamefont {Langer},\ and\ \citenamefont
  {Wineland}}]{leibfried2004toward}%
  \BibitemOpen
  \bibfield  {author} {\bibinfo {author} {\bibfnamefont {D.}~\bibnamefont
  {Leibfried}}, \bibinfo {author} {\bibfnamefont {M.}~\bibnamefont {Barrett}},
  \bibinfo {author} {\bibfnamefont {T.}~\bibnamefont {Schaetz}}, \bibinfo
  {author} {\bibfnamefont {J.}~\bibnamefont {Britton}}, \bibinfo {author}
  {\bibfnamefont {J.}~\bibnamefont {Chiaverini}}, \bibinfo {author}
  {\bibfnamefont {W.}~\bibnamefont {Itano}}, \bibinfo {author} {\bibfnamefont
  {J.}~\bibnamefont {Jost}}, \bibinfo {author} {\bibfnamefont {C.}~\bibnamefont
  {Langer}}, \ and\ \bibinfo {author} {\bibfnamefont {D.}~\bibnamefont
  {Wineland}},\ }\bibfield  {title} {\emph {\bibinfo {title} {Toward
  heisenberg-limited spectroscopy with multiparticle entangled states},\
  }}\href {https://doi.org/10.1126/science.1097576} {\bibfield  {journal}
  {\bibinfo  {journal} {Science}\ }\textbf {\bibinfo {volume} {304}},\ \bibinfo
  {pages} {1476} (\bibinfo {year} {2004})}\BibitemShut {NoStop}%
\bibitem [{\citenamefont {Giovannetti}\ \emph {et~al.}(2004)\citenamefont
  {Giovannetti}, \citenamefont {Lloyd},\ and\ \citenamefont
  {Maccone}}]{giovannetti2004quantum}%
  \BibitemOpen
  \bibfield  {author} {\bibinfo {author} {\bibfnamefont {V.}~\bibnamefont
  {Giovannetti}}, \bibinfo {author} {\bibfnamefont {S.}~\bibnamefont {Lloyd}},
  \ and\ \bibinfo {author} {\bibfnamefont {L.}~\bibnamefont {Maccone}},\
  }\bibfield  {title} {\emph {\bibinfo {title} {Quantum-enhanced measurements:
  beating the standard quantum limit},\ }}\href
  {http://doi.org/10.1126/science.1104149} {\bibfield  {journal} {\bibinfo
  {journal} {Science}\ }\textbf {\bibinfo {volume} {306}},\ \bibinfo {pages}
  {1330} (\bibinfo {year} {2004})}\BibitemShut {NoStop}%
\bibitem [{\citenamefont {Huelga}\ \emph {et~al.}(1997)\citenamefont {Huelga},
  \citenamefont {Macchiavello}, \citenamefont {Pellizzari}, \citenamefont
  {Ekert}, \citenamefont {Plenio},\ and\ \citenamefont
  {Cirac}}]{huelga1997improvement}%
  \BibitemOpen
  \bibfield  {author} {\bibinfo {author} {\bibfnamefont {S.~F.}\ \bibnamefont
  {Huelga}}, \bibinfo {author} {\bibfnamefont {C.}~\bibnamefont
  {Macchiavello}}, \bibinfo {author} {\bibfnamefont {T.}~\bibnamefont
  {Pellizzari}}, \bibinfo {author} {\bibfnamefont {A.~K.}\ \bibnamefont
  {Ekert}}, \bibinfo {author} {\bibfnamefont {M.~B.}\ \bibnamefont {Plenio}}, \
  and\ \bibinfo {author} {\bibfnamefont {J.~I.}\ \bibnamefont {Cirac}},\
  }\bibfield  {title} {\emph {\bibinfo {title} {Improvement of frequency
  standards with quantum entanglement},\ }}\href
  {https://doi.org/10.1103/PhysRevLett.79.3865} {\bibfield  {journal} {\bibinfo
   {journal} {Phys. Rev. Lett.}\ }\textbf {\bibinfo {volume} {79}},\ \bibinfo
  {pages} {3865} (\bibinfo {year} {1997})}\BibitemShut {NoStop}%
\bibitem [{\citenamefont {Berry}\ and\ \citenamefont
  {Wiseman}(2000)}]{Berry2000}%
  \BibitemOpen
  \bibfield  {author} {\bibinfo {author} {\bibfnamefont {D.~W.}\ \bibnamefont
  {Berry}}\ and\ \bibinfo {author} {\bibfnamefont {H.~M.}\ \bibnamefont
  {Wiseman}},\ }\bibfield  {title} {\emph {\bibinfo {title} {Optimal states and
  almost optimal adaptive measurements for quantum interferometry},\ }}\href
  {https://doi.org/10.1103/PhysRevLett.85.5098} {\bibfield  {journal} {\bibinfo
   {journal} {Phys. Rev. Lett.}\ }\textbf {\bibinfo {volume} {85}},\ \bibinfo
  {pages} {5098} (\bibinfo {year} {2000})}\BibitemShut {NoStop}%
\bibitem [{\citenamefont {de~Burgh}\ and\ \citenamefont
  {Bartlett}(2005)}]{de2005quantum}%
  \BibitemOpen
  \bibfield  {author} {\bibinfo {author} {\bibfnamefont {M.}~\bibnamefont
  {de~Burgh}}\ and\ \bibinfo {author} {\bibfnamefont {S.~D.}\ \bibnamefont
  {Bartlett}},\ }\bibfield  {title} {\emph {\bibinfo {title} {Quantum methods
  for clock synchronization: Beating the standard quantum limit without
  entanglement},\ }}\href {https://doi.org/10.1103/PhysRevA.72.042301}
  {\bibfield  {journal} {\bibinfo  {journal} {Phys. Rev. A}\ }\textbf {\bibinfo
  {volume} {72}},\ \bibinfo {pages} {042301} (\bibinfo {year}
  {2005})}\BibitemShut {NoStop}%
\bibitem [{\citenamefont {Fujiwara}\ and\ \citenamefont
  {Imai}(2008)}]{fujiwara2008fibre}%
  \BibitemOpen
  \bibfield  {author} {\bibinfo {author} {\bibfnamefont {A.}~\bibnamefont
  {Fujiwara}}\ and\ \bibinfo {author} {\bibfnamefont {H.}~\bibnamefont
  {Imai}},\ }\bibfield  {title} {\emph {\bibinfo {title} {A fibre bundle over
  manifolds of quantum channels and its application to quantum statistics},\
  }}\href {https://doi.org/10.1088/1751-8113/41/25/255304} {\bibfield
  {journal} {\bibinfo  {journal} {J. Phys. A: Math. Theor.}\ }\textbf {\bibinfo
  {volume} {41}},\ \bibinfo {pages} {255304} (\bibinfo {year}
  {2008})}\BibitemShut {NoStop}%
\bibitem [{\citenamefont {Demkowicz-Dobrza{\'n}ski}\ \emph
  {et~al.}(2009)\citenamefont {Demkowicz-Dobrza{\'n}ski}, \citenamefont
  {Dorner}, \citenamefont {Smith}, \citenamefont {Lundeen}, \citenamefont
  {Wasilewski}, \citenamefont {Banaszek},\ and\ \citenamefont
  {Walmsley}}]{demkowicz2009quantum}%
  \BibitemOpen
  \bibfield  {author} {\bibinfo {author} {\bibfnamefont {R.}~\bibnamefont
  {Demkowicz-Dobrza{\'n}ski}}, \bibinfo {author} {\bibfnamefont
  {U.}~\bibnamefont {Dorner}}, \bibinfo {author} {\bibfnamefont
  {B.}~\bibnamefont {Smith}}, \bibinfo {author} {\bibfnamefont
  {J.}~\bibnamefont {Lundeen}}, \bibinfo {author} {\bibfnamefont
  {W.}~\bibnamefont {Wasilewski}}, \bibinfo {author} {\bibfnamefont
  {K.}~\bibnamefont {Banaszek}}, \ and\ \bibinfo {author} {\bibfnamefont
  {I.}~\bibnamefont {Walmsley}},\ }\bibfield  {title} {\emph {\bibinfo {title}
  {Quantum phase estimation with lossy interferometers},\ }}\href
  {https://doi.org/10.1103/PhysRevA.80.013825} {\bibfield  {journal} {\bibinfo
  {journal} {Phys. Rev. A}\ }\textbf {\bibinfo {volume} {80}},\ \bibinfo
  {pages} {013825} (\bibinfo {year} {2009})}\BibitemShut {NoStop}%
\bibitem [{\citenamefont {Escher}\ \emph {et~al.}(2011)\citenamefont {Escher},
  \citenamefont {de~Matos~Filho},\ and\ \citenamefont
  {Davidovich}}]{escher2011general}%
  \BibitemOpen
  \bibfield  {author} {\bibinfo {author} {\bibfnamefont {B.}~\bibnamefont
  {Escher}}, \bibinfo {author} {\bibfnamefont {R.}~\bibnamefont
  {de~Matos~Filho}}, \ and\ \bibinfo {author} {\bibfnamefont {L.}~\bibnamefont
  {Davidovich}},\ }\bibfield  {title} {\emph {\bibinfo {title} {General
  framework for estimating the ultimate precision limit in noisy
  quantum-enhanced metrology},\ }}\href {https://doi.org/10.1038/nphys1958}
  {\bibfield  {journal} {\bibinfo  {journal} {Nat. Phys.}\ }\textbf {\bibinfo
  {volume} {7}},\ \bibinfo {pages} {406} (\bibinfo {year} {2011})}\BibitemShut
  {NoStop}%
\bibitem [{\citenamefont {Demkowicz-Dobrza{\'n}ski}\ \emph
  {et~al.}(2012)\citenamefont {Demkowicz-Dobrza{\'n}ski}, \citenamefont
  {Ko{\l}ody{\'n}ski},\ and\ \citenamefont
  {Gu{\c{t}}{\u{a}}}}]{demkowicz2012elusive}%
  \BibitemOpen
  \bibfield  {author} {\bibinfo {author} {\bibfnamefont {R.}~\bibnamefont
  {Demkowicz-Dobrza{\'n}ski}}, \bibinfo {author} {\bibfnamefont
  {J.}~\bibnamefont {Ko{\l}ody{\'n}ski}}, \ and\ \bibinfo {author}
  {\bibfnamefont {M.}~\bibnamefont {Gu{\c{t}}{\u{a}}}},\ }\bibfield  {title}
  {\emph {\bibinfo {title} {The elusive heisenberg limit in quantum-enhanced
  metrology},\ }}\href {https://doi.org/10.1038/ncomms2067} {\bibfield
  {journal} {\bibinfo  {journal} {Nat. Commun.}\ }\textbf {\bibinfo {volume}
  {3}},\ \bibinfo {pages} {1063} (\bibinfo {year} {2012})}\BibitemShut
  {NoStop}%
\bibitem [{\citenamefont {Ko{\l}ody{\'n}ski}\ and\ \citenamefont
  {Demkowicz-Dobrza{\'n}ski}(2013)}]{kolodynski2013efficient}%
  \BibitemOpen
  \bibfield  {author} {\bibinfo {author} {\bibfnamefont {J.}~\bibnamefont
  {Ko{\l}ody{\'n}ski}}\ and\ \bibinfo {author} {\bibfnamefont {R.}~\bibnamefont
  {Demkowicz-Dobrza{\'n}ski}},\ }\bibfield  {title} {\emph {\bibinfo {title}
  {Efficient tools for quantum metrology with uncorrelated noise},\ }}\href
  {http://doi.org/10.1088/1367-2630/15/7/073043} {\bibfield  {journal}
  {\bibinfo  {journal} {New J. Phys.}\ }\textbf {\bibinfo {volume} {15}},\
  \bibinfo {pages} {073043} (\bibinfo {year} {2013})}\BibitemShut {NoStop}%
\bibitem [{\citenamefont {Knysh}\ \emph {et~al.}(2014)\citenamefont {Knysh},
  \citenamefont {Chen},\ and\ \citenamefont {Durkin}}]{knysh2014true}%
  \BibitemOpen
  \bibfield  {author} {\bibinfo {author} {\bibfnamefont {S.~I.}\ \bibnamefont
  {Knysh}}, \bibinfo {author} {\bibfnamefont {E.~H.}\ \bibnamefont {Chen}}, \
  and\ \bibinfo {author} {\bibfnamefont {G.~A.}\ \bibnamefont {Durkin}},\
  }\bibfield  {title} {\emph {\bibinfo {title} {True limits to precision via
  unique quantum probe},\ }}\href {https://arxiv.org/pdf/1402.0495} {\bibfield
  {journal} {\bibinfo  {journal} {arXiv:1402.0495}\ } (\bibinfo {year}
  {2014})}\BibitemShut {NoStop}%
\bibitem [{\citenamefont {Demkowicz-Dobrza{\'n}ski}\ and\ \citenamefont
  {Maccone}(2014)}]{demkowicz2014using}%
  \BibitemOpen
  \bibfield  {author} {\bibinfo {author} {\bibfnamefont {R.}~\bibnamefont
  {Demkowicz-Dobrza{\'n}ski}}\ and\ \bibinfo {author} {\bibfnamefont
  {L.}~\bibnamefont {Maccone}},\ }\bibfield  {title} {\emph {\bibinfo {title}
  {Using entanglement against noise in quantum metrology},\ }}\href
  {https://doi.org/10.1103/PhysRevLett.113.250801} {\bibfield  {journal}
  {\bibinfo  {journal} {Phys. Rev. Lett.}\ }\textbf {\bibinfo {volume} {113}},\
  \bibinfo {pages} {250801} (\bibinfo {year} {2014})}\BibitemShut {NoStop}%
\bibitem [{\citenamefont {Kessler}\ \emph {et~al.}(2014)\citenamefont
  {Kessler}, \citenamefont {Lovchinsky}, \citenamefont {Sushkov},\ and\
  \citenamefont {Lukin}}]{kessler2014quantum}%
  \BibitemOpen
  \bibfield  {author} {\bibinfo {author} {\bibfnamefont {E.~M.}\ \bibnamefont
  {Kessler}}, \bibinfo {author} {\bibfnamefont {I.}~\bibnamefont {Lovchinsky}},
  \bibinfo {author} {\bibfnamefont {A.~O.}\ \bibnamefont {Sushkov}}, \ and\
  \bibinfo {author} {\bibfnamefont {M.~D.}\ \bibnamefont {Lukin}},\ }\bibfield
  {title} {\emph {\bibinfo {title} {Quantum error correction for metrology},\
  }}\href {https://doi.org/10.1103/PhysRevLett.112.150802} {\bibfield
  {journal} {\bibinfo  {journal} {Phys. Rev. Lett.}\ }\textbf {\bibinfo
  {volume} {112}},\ \bibinfo {pages} {150802} (\bibinfo {year}
  {2014})}\BibitemShut {NoStop}%
\bibitem [{\citenamefont {D{\"u}r}\ \emph {et~al.}(2014)\citenamefont
  {D{\"u}r}, \citenamefont {Skotiniotis}, \citenamefont {Froewis},\ and\
  \citenamefont {Kraus}}]{dur2014improved}%
  \BibitemOpen
  \bibfield  {author} {\bibinfo {author} {\bibfnamefont {W.}~\bibnamefont
  {D{\"u}r}}, \bibinfo {author} {\bibfnamefont {M.}~\bibnamefont
  {Skotiniotis}}, \bibinfo {author} {\bibfnamefont {F.}~\bibnamefont
  {Froewis}}, \ and\ \bibinfo {author} {\bibfnamefont {B.}~\bibnamefont
  {Kraus}},\ }\bibfield  {title} {\emph {\bibinfo {title} {Improved quantum
  metrology using quantum error correction},\ }}\href
  {https://doi.org/10.1103/PhysRevLett.112.080801} {\bibfield  {journal}
  {\bibinfo  {journal} {Phys. Rev. Lett.}\ }\textbf {\bibinfo {volume} {112}},\
  \bibinfo {pages} {080801} (\bibinfo {year} {2014})}\BibitemShut {NoStop}%
\bibitem [{\citenamefont {Ozeri}(2013)}]{ozeri2013heisenberg}%
  \BibitemOpen
  \bibfield  {author} {\bibinfo {author} {\bibfnamefont {R.}~\bibnamefont
  {Ozeri}},\ }\bibfield  {title} {\emph {\bibinfo {title} {Heisenberg limited
  metrology using quantum error-correction codes.}\ }}\href
  {https://arxiv.org/pdf/1310.3432} {\bibfield  {journal} {\bibinfo  {journal}
  {arXiv:1310.3432}\ } (\bibinfo {year} {2013})}\BibitemShut {NoStop}%
\bibitem [{\citenamefont {Arrad}\ \emph {et~al.}(2014)\citenamefont {Arrad},
  \citenamefont {Vinkler}, \citenamefont {Aharonov},\ and\ \citenamefont
  {Retzker}}]{arrad2014increasing}%
  \BibitemOpen
  \bibfield  {author} {\bibinfo {author} {\bibfnamefont {G.}~\bibnamefont
  {Arrad}}, \bibinfo {author} {\bibfnamefont {Y.}~\bibnamefont {Vinkler}},
  \bibinfo {author} {\bibfnamefont {D.}~\bibnamefont {Aharonov}}, \ and\
  \bibinfo {author} {\bibfnamefont {A.}~\bibnamefont {Retzker}},\ }\bibfield
  {title} {\emph {\bibinfo {title} {Increasing sensing resolution with error
  correction},\ }}\href {https://doi.org/10.1103/PhysRevLett.112.150801}
  {\bibfield  {journal} {\bibinfo  {journal} {Phys. Rev. Lett.}\ }\textbf
  {\bibinfo {volume} {112}},\ \bibinfo {pages} {150801} (\bibinfo {year}
  {2014})}\BibitemShut {NoStop}%
\bibitem [{\citenamefont {Unden}\ \emph {et~al.}(2016)\citenamefont {Unden},
  \citenamefont {Balasubramanian}, \citenamefont {Louzon}, \citenamefont
  {Vinkler}, \citenamefont {Plenio}, \citenamefont {Markham}, \citenamefont
  {Twitchen}, \citenamefont {Stacey}, \citenamefont {Lovchinsky}, \citenamefont
  {Sushkov} \emph {et~al.}}]{unden2016quantum}%
  \BibitemOpen
  \bibfield  {author} {\bibinfo {author} {\bibfnamefont {T.}~\bibnamefont
  {Unden}}, \bibinfo {author} {\bibfnamefont {P.}~\bibnamefont
  {Balasubramanian}}, \bibinfo {author} {\bibfnamefont {D.}~\bibnamefont
  {Louzon}}, \bibinfo {author} {\bibfnamefont {Y.}~\bibnamefont {Vinkler}},
  \bibinfo {author} {\bibfnamefont {M.~B.}\ \bibnamefont {Plenio}}, \bibinfo
  {author} {\bibfnamefont {M.}~\bibnamefont {Markham}}, \bibinfo {author}
  {\bibfnamefont {D.}~\bibnamefont {Twitchen}}, \bibinfo {author}
  {\bibfnamefont {A.}~\bibnamefont {Stacey}}, \bibinfo {author} {\bibfnamefont
  {I.}~\bibnamefont {Lovchinsky}}, \bibinfo {author} {\bibfnamefont {A.~O.}\
  \bibnamefont {Sushkov}},  \emph {et~al.},\ }\bibfield  {title} {\emph
  {\bibinfo {title} {Quantum metrology enhanced by repetitive quantum error
  correction},\ }}\href {https://doi.org/10.1103/PhysRevLett.116.230502}
  {\bibfield  {journal} {\bibinfo  {journal} {Phys. Rev. Lett.}\ }\textbf
  {\bibinfo {volume} {116}},\ \bibinfo {pages} {230502} (\bibinfo {year}
  {2016})}\BibitemShut {NoStop}%
\bibitem [{\citenamefont {Reiter}\ \emph {et~al.}(2017)\citenamefont {Reiter},
  \citenamefont {S{\o}rensen}, \citenamefont {Zoller},\ and\ \citenamefont
  {Muschik}}]{reiter2017dissipative}%
  \BibitemOpen
  \bibfield  {author} {\bibinfo {author} {\bibfnamefont {F.}~\bibnamefont
  {Reiter}}, \bibinfo {author} {\bibfnamefont {A.~S.}\ \bibnamefont
  {S{\o}rensen}}, \bibinfo {author} {\bibfnamefont {P.}~\bibnamefont {Zoller}},
  \ and\ \bibinfo {author} {\bibfnamefont {C.~A.}\ \bibnamefont {Muschik}},\
  }\bibfield  {title} {\emph {\bibinfo {title} {Dissipative quantum error
  correction and application to quantum sensing with trapped ions},\ }}\href
  {https://doi.org/10.1038/s41467-017-01895-5} {\bibfield  {journal} {\bibinfo
  {journal} {Nat. Commun.}\ }\textbf {\bibinfo {volume} {8}},\ \bibinfo {pages}
  {1822} (\bibinfo {year} {2017})}\BibitemShut {NoStop}%
\bibitem [{\citenamefont {Sekatski}\ \emph {et~al.}(2017)\citenamefont
  {Sekatski}, \citenamefont {Skotiniotis}, \citenamefont {Ko{\l}ody{\'n}ski},\
  and\ \citenamefont {D{\"u}r}}]{sekatski2017quantum}%
  \BibitemOpen
  \bibfield  {author} {\bibinfo {author} {\bibfnamefont {P.}~\bibnamefont
  {Sekatski}}, \bibinfo {author} {\bibfnamefont {M.}~\bibnamefont
  {Skotiniotis}}, \bibinfo {author} {\bibfnamefont {J.}~\bibnamefont
  {Ko{\l}ody{\'n}ski}}, \ and\ \bibinfo {author} {\bibfnamefont
  {W.}~\bibnamefont {D{\"u}r}},\ }\bibfield  {title} {\emph {\bibinfo {title}
  {Quantum metrology with full and fast quantum control},\ }}\href
  {https://doi.org/10.22331/q-2017-09-06-27} {\bibfield  {journal} {\bibinfo
  {journal} {Quantum}\ }\textbf {\bibinfo {volume} {1}},\ \bibinfo {pages} {27}
  (\bibinfo {year} {2017})}\BibitemShut {NoStop}%
\bibitem [{\citenamefont {Demkowicz-Dobrza\ifmmode~\acute{n}\else
  \'{n}\fi{}ski}\ \emph {et~al.}(2017)\citenamefont
  {Demkowicz-Dobrza\ifmmode~\acute{n}\else \'{n}\fi{}ski}, \citenamefont
  {Czajkowski},\ and\ \citenamefont {Sekatski}}]{demkowicz2017adaptive}%
  \BibitemOpen
  \bibfield  {author} {\bibinfo {author} {\bibfnamefont {R.}~\bibnamefont
  {Demkowicz-Dobrza\ifmmode~\acute{n}\else \'{n}\fi{}ski}}, \bibinfo {author}
  {\bibfnamefont {J.}~\bibnamefont {Czajkowski}}, \ and\ \bibinfo {author}
  {\bibfnamefont {P.}~\bibnamefont {Sekatski}},\ }\bibfield  {title} {\emph
  {\bibinfo {title} {Adaptive quantum metrology under general markovian
  noise},\ }}\href {https://doi.org/10.1103/PhysRevX.7.041009} {\bibfield
  {journal} {\bibinfo  {journal} {Phys. Rev. X}\ }\textbf {\bibinfo {volume}
  {7}},\ \bibinfo {pages} {041009} (\bibinfo {year} {2017})}\BibitemShut
  {NoStop}%
\bibitem [{\citenamefont {Zhou}\ \emph {et~al.}(2018)\citenamefont {Zhou},
  \citenamefont {Zhang}, \citenamefont {Preskill},\ and\ \citenamefont
  {Jiang}}]{zhou2018achieving}%
  \BibitemOpen
  \bibfield  {author} {\bibinfo {author} {\bibfnamefont {S.}~\bibnamefont
  {Zhou}}, \bibinfo {author} {\bibfnamefont {M.}~\bibnamefont {Zhang}},
  \bibinfo {author} {\bibfnamefont {J.}~\bibnamefont {Preskill}}, \ and\
  \bibinfo {author} {\bibfnamefont {L.}~\bibnamefont {Jiang}},\ }\bibfield
  {title} {\emph {\bibinfo {title} {Achieving the heisenberg limit in quantum
  metrology using quantum error correction},\ }}\href
  {https://doi.org/10.1038/s41467-017-02510-3} {\bibfield  {journal} {\bibinfo
  {journal} {Nat. Commun.}\ }\textbf {\bibinfo {volume} {9}},\ \bibinfo {pages}
  {78} (\bibinfo {year} {2018})}\BibitemShut {NoStop}%
\bibitem [{\citenamefont {Layden}\ and\ \citenamefont
  {Cappellaro}(2018)}]{layden2018spatial}%
  \BibitemOpen
  \bibfield  {author} {\bibinfo {author} {\bibfnamefont {D.}~\bibnamefont
  {Layden}}\ and\ \bibinfo {author} {\bibfnamefont {P.}~\bibnamefont
  {Cappellaro}},\ }\bibfield  {title} {\emph {\bibinfo {title} {Spatial noise
  filtering through error correction for quantum sensing},\ }}\href
  {https://doi.org/10.1038/s41534-018-0082-2} {\bibfield  {journal} {\bibinfo
  {journal} {npj Quantum Inf.}\ }\textbf {\bibinfo {volume} {4}},\ \bibinfo
  {pages} {30} (\bibinfo {year} {2018})}\BibitemShut {NoStop}%
\bibitem [{\citenamefont {Layden}\ \emph {et~al.}(2019)\citenamefont {Layden},
  \citenamefont {Zhou}, \citenamefont {Cappellaro},\ and\ \citenamefont
  {Jiang}}]{layden2018ancilla}%
  \BibitemOpen
  \bibfield  {author} {\bibinfo {author} {\bibfnamefont {D.}~\bibnamefont
  {Layden}}, \bibinfo {author} {\bibfnamefont {S.}~\bibnamefont {Zhou}},
  \bibinfo {author} {\bibfnamefont {P.}~\bibnamefont {Cappellaro}}, \ and\
  \bibinfo {author} {\bibfnamefont {L.}~\bibnamefont {Jiang}},\ }\bibfield
  {title} {\emph {\bibinfo {title} {Ancilla-free quantum error correction codes
  for quantum metrology},\ }}\href
  {https://doi.org/10.1103/PhysRevLett.122.040502} {\bibfield  {journal}
  {\bibinfo  {journal} {Phys. Rev. Lett.}\ }\textbf {\bibinfo {volume} {122}},\
  \bibinfo {pages} {040502} (\bibinfo {year} {2019})}\BibitemShut {NoStop}%
\bibitem [{\citenamefont {Kapourniotis}\ and\ \citenamefont
  {Datta}(2019)}]{kapourniotis2019fault}%
  \BibitemOpen
  \bibfield  {author} {\bibinfo {author} {\bibfnamefont {T.}~\bibnamefont
  {Kapourniotis}}\ and\ \bibinfo {author} {\bibfnamefont {A.}~\bibnamefont
  {Datta}},\ }\bibfield  {title} {\emph {\bibinfo {title} {Fault-tolerant
  quantum metrology},\ }}\href {https://doi.org/10.1103/PhysRevA.100.022335}
  {\bibfield  {journal} {\bibinfo  {journal} {Phys. Rev. A}\ }\textbf {\bibinfo
  {volume} {100}},\ \bibinfo {pages} {022335} (\bibinfo {year}
  {2019})}\BibitemShut {NoStop}%
\bibitem [{\citenamefont {Tan}\ \emph {et~al.}(2019)\citenamefont {Tan},
  \citenamefont {Omkar},\ and\ \citenamefont {Jeong}}]{tan2019quantum}%
  \BibitemOpen
  \bibfield  {author} {\bibinfo {author} {\bibfnamefont {K.~C.}\ \bibnamefont
  {Tan}}, \bibinfo {author} {\bibfnamefont {S.}~\bibnamefont {Omkar}}, \ and\
  \bibinfo {author} {\bibfnamefont {H.}~\bibnamefont {Jeong}},\ }\bibfield
  {title} {\emph {\bibinfo {title} {Quantum-error-correction-assisted quantum
  metrology without entanglement},\ }}\href
  {https://doi.org/10.1103/PhysRevA.100.022312} {\bibfield  {journal} {\bibinfo
   {journal} {Phys. Rev. A}\ }\textbf {\bibinfo {volume} {100}},\ \bibinfo
  {pages} {022312} (\bibinfo {year} {2019})}\BibitemShut {NoStop}%
\bibitem [{\citenamefont {Zhou}\ and\ \citenamefont
  {Jiang}(2020{\natexlab{a}})}]{zhou2020theory}%
  \BibitemOpen
  \bibfield  {author} {\bibinfo {author} {\bibfnamefont {S.}~\bibnamefont
  {Zhou}}\ and\ \bibinfo {author} {\bibfnamefont {L.}~\bibnamefont {Jiang}},\
  }\bibfield  {title} {\emph {\bibinfo {title} {The theory of
  entanglement-assisted metrology for quantum channels},\ }}\href
  {https://arxiv.org/abs/2003.10559} {\bibfield  {journal} {\bibinfo  {journal}
  {arXiv:2003.10559}\ } (\bibinfo {year} {2020}{\natexlab{a}})}\BibitemShut
  {NoStop}%
\bibitem [{\citenamefont {Chen}\ \emph {et~al.}(2020)\citenamefont {Chen},
  \citenamefont {Chen}, \citenamefont {Liu}, \citenamefont {Miao},\ and\
  \citenamefont {Yuan}}]{chen2020fluctuation}%
  \BibitemOpen
  \bibfield  {author} {\bibinfo {author} {\bibfnamefont {Y.}~\bibnamefont
  {Chen}}, \bibinfo {author} {\bibfnamefont {H.}~\bibnamefont {Chen}}, \bibinfo
  {author} {\bibfnamefont {J.}~\bibnamefont {Liu}}, \bibinfo {author}
  {\bibfnamefont {Z.}~\bibnamefont {Miao}}, \ and\ \bibinfo {author}
  {\bibfnamefont {H.}~\bibnamefont {Yuan}},\ }\bibfield  {title} {\emph
  {\bibinfo {title} {Fluctuation-enhanced quantum metrology},\ }}\href
  {https://arxiv.org/pdf/2003.13010} {\bibfield  {journal} {\bibinfo  {journal}
  {arXiv:2003.13010}\ } (\bibinfo {year} {2020})}\BibitemShut {NoStop}%
\bibitem [{\citenamefont {Baumgratz}\ and\ \citenamefont
  {Datta}(2016)}]{baumgratz2016quantum}%
  \BibitemOpen
  \bibfield  {author} {\bibinfo {author} {\bibfnamefont {T.}~\bibnamefont
  {Baumgratz}}\ and\ \bibinfo {author} {\bibfnamefont {A.}~\bibnamefont
  {Datta}},\ }\bibfield  {title} {\emph {\bibinfo {title} {Quantum enhanced
  estimation of a multidimensional field},\ }}\href
  {https://doi.org/10.1103/PhysRevLett.116.030801} {\bibfield  {journal}
  {\bibinfo  {journal} {Phys. Rev. Lett.}\ }\textbf {\bibinfo {volume} {116}},\
  \bibinfo {pages} {030801} (\bibinfo {year} {2016})}\BibitemShut {NoStop}%
\bibitem [{\citenamefont {Tsang}\ \emph {et~al.}(2016)\citenamefont {Tsang},
  \citenamefont {Nair},\ and\ \citenamefont {Lu}}]{tsang2016quantum}%
  \BibitemOpen
  \bibfield  {author} {\bibinfo {author} {\bibfnamefont {M.}~\bibnamefont
  {Tsang}}, \bibinfo {author} {\bibfnamefont {R.}~\bibnamefont {Nair}}, \ and\
  \bibinfo {author} {\bibfnamefont {X.-M.}\ \bibnamefont {Lu}},\ }\bibfield
  {title} {\emph {\bibinfo {title} {Quantum theory of superresolution for two
  incoherent optical point sources},\ }}\href
  {https://doi.org/10.1103/PhysRevX.6.031033} {\bibfield  {journal} {\bibinfo
  {journal} {Phys. Rev. X}\ }\textbf {\bibinfo {volume} {6}},\ \bibinfo {pages}
  {031033} (\bibinfo {year} {2016})}\BibitemShut {NoStop}%
\bibitem [{\citenamefont {Humphreys}\ \emph {et~al.}(2013)\citenamefont
  {Humphreys}, \citenamefont {Barbieri}, \citenamefont {Datta},\ and\
  \citenamefont {Walmsley}}]{humphreys2013quantum}%
  \BibitemOpen
  \bibfield  {author} {\bibinfo {author} {\bibfnamefont {P.~C.}\ \bibnamefont
  {Humphreys}}, \bibinfo {author} {\bibfnamefont {M.}~\bibnamefont {Barbieri}},
  \bibinfo {author} {\bibfnamefont {A.}~\bibnamefont {Datta}}, \ and\ \bibinfo
  {author} {\bibfnamefont {I.~A.}\ \bibnamefont {Walmsley}},\ }\bibfield
  {title} {\emph {\bibinfo {title} {Quantum enhanced multiple phase
  estimation},\ }}\href {https://doi.org/10.1103/PhysRevLett.111.070403}
  {\bibfield  {journal} {\bibinfo  {journal} {Phys. Rev. Lett.}\ }\textbf
  {\bibinfo {volume} {111}},\ \bibinfo {pages} {070403} (\bibinfo {year}
  {2013})}\BibitemShut {NoStop}%
\bibitem [{\citenamefont {Gessner}\ \emph {et~al.}(2018)\citenamefont
  {Gessner}, \citenamefont {Pezz\`e},\ and\ \citenamefont
  {Smerzi}}]{Gessner2018}%
  \BibitemOpen
  \bibfield  {author} {\bibinfo {author} {\bibfnamefont {M.}~\bibnamefont
  {Gessner}}, \bibinfo {author} {\bibfnamefont {L.}~\bibnamefont {Pezz\`e}}, \
  and\ \bibinfo {author} {\bibfnamefont {A.}~\bibnamefont {Smerzi}},\
  }\bibfield  {title} {\emph {\bibinfo {title} {Sensitivity bounds for
  multiparameter quantum metrology},\ }}\href
  {https://doi.org/10.1103/PhysRevLett.121.130503} {\bibfield  {journal}
  {\bibinfo  {journal} {Phys. Rev. Lett.}\ }\textbf {\bibinfo {volume} {121}},\
  \bibinfo {pages} {130503} (\bibinfo {year} {2018})}\BibitemShut {NoStop}%
\bibitem [{\citenamefont {Tsang}\ \emph {et~al.}(2011)\citenamefont {Tsang},
  \citenamefont {Wiseman},\ and\ \citenamefont {Caves}}]{Tsang2011}%
  \BibitemOpen
  \bibfield  {author} {\bibinfo {author} {\bibfnamefont {M.}~\bibnamefont
  {Tsang}}, \bibinfo {author} {\bibfnamefont {H.~M.}\ \bibnamefont {Wiseman}},
  \ and\ \bibinfo {author} {\bibfnamefont {C.~M.}\ \bibnamefont {Caves}},\
  }\bibfield  {title} {\emph {\bibinfo {title} {Fundamental quantum limit to
  waveform estimation},\ }}\href
  {http://doi.org/10.1103/PhysRevLett.106.090401} {\bibfield  {journal}
  {\bibinfo  {journal} {Phys. Rev. Lett.}\ }\textbf {\bibinfo {volume} {106}},\
  \bibinfo {pages} {090401} (\bibinfo {year} {2011})}\BibitemShut {NoStop}%
\bibitem [{\citenamefont {Berry}\ \emph {et~al.}(2013)\citenamefont {Berry},
  \citenamefont {Hall},\ and\ \citenamefont {Wiseman}}]{Berry2013}%
  \BibitemOpen
  \bibfield  {author} {\bibinfo {author} {\bibfnamefont {D.~W.}\ \bibnamefont
  {Berry}}, \bibinfo {author} {\bibfnamefont {M.~J.~W.}\ \bibnamefont {Hall}},
  \ and\ \bibinfo {author} {\bibfnamefont {H.~M.}\ \bibnamefont {Wiseman}},\
  }\bibfield  {title} {\emph {\bibinfo {title} {Stochastic heisenberg limit:
  Optimal estimation of a fluctuating phase},\ }}\href
  {https://doi.org/10.1103/PhysRevLett.111.113601} {\bibfield  {journal}
  {\bibinfo  {journal} {Phys. Rev. Lett.}\ }\textbf {\bibinfo {volume} {111}},\
  \bibinfo {pages} {113601} (\bibinfo {year} {2013})}\BibitemShut {NoStop}%
\bibitem [{\citenamefont {Matsumoto}(2002)}]{matsumoto2002new}%
  \BibitemOpen
  \bibfield  {author} {\bibinfo {author} {\bibfnamefont {K.}~\bibnamefont
  {Matsumoto}},\ }\bibfield  {title} {\emph {\bibinfo {title} {A new approach
  to the cram{\'e}r-rao-type bound of the pure-state model},\ }}\href
  {http://doi.org/10.1088/0305-4470/35/13/307} {\bibfield  {journal} {\bibinfo
  {journal} {J. Phys. A.: Math. Theor.}\ }\textbf {\bibinfo {volume} {35}},\
  \bibinfo {pages} {3111} (\bibinfo {year} {2002})}\BibitemShut {NoStop}%
\bibitem [{\citenamefont {Genoni}\ \emph {et~al.}(2013)\citenamefont {Genoni},
  \citenamefont {Paris}, \citenamefont {Adesso}, \citenamefont {Nha},
  \citenamefont {Knight},\ and\ \citenamefont {Kim}}]{genoni2013optimal}%
  \BibitemOpen
  \bibfield  {author} {\bibinfo {author} {\bibfnamefont {M.~G.}\ \bibnamefont
  {Genoni}}, \bibinfo {author} {\bibfnamefont {M.~G.~A.}\ \bibnamefont
  {Paris}}, \bibinfo {author} {\bibfnamefont {G.}~\bibnamefont {Adesso}},
  \bibinfo {author} {\bibfnamefont {H.}~\bibnamefont {Nha}}, \bibinfo {author}
  {\bibfnamefont {P.~L.}\ \bibnamefont {Knight}}, \ and\ \bibinfo {author}
  {\bibfnamefont {M.~S.}\ \bibnamefont {Kim}},\ }\bibfield  {title} {\emph
  {\bibinfo {title} {Optimal estimation of joint parameters in phase space},\
  }}\href {\doibase 10.1103/PhysRevA.87.012107} {\bibfield  {journal} {\bibinfo
   {journal} {Phys. Rev. A}\ }\textbf {\bibinfo {volume} {87}},\ \bibinfo
  {pages} {012107} (\bibinfo {year} {2013})}\BibitemShut {NoStop}%
\bibitem [{\citenamefont {Ragy}\ \emph {et~al.}(2016)\citenamefont {Ragy},
  \citenamefont {Jarzyna},\ and\ \citenamefont
  {Demkowicz-Dobrza\ifmmode~\acute{n}\else
  \'{n}\fi{}ski}}]{ragy2016compatibility}%
  \BibitemOpen
  \bibfield  {author} {\bibinfo {author} {\bibfnamefont {S.}~\bibnamefont
  {Ragy}}, \bibinfo {author} {\bibfnamefont {M.}~\bibnamefont {Jarzyna}}, \
  and\ \bibinfo {author} {\bibfnamefont {R.}~\bibnamefont
  {Demkowicz-Dobrza\ifmmode~\acute{n}\else \'{n}\fi{}ski}},\ }\bibfield
  {title} {\emph {\bibinfo {title} {Compatibility in multiparameter quantum
  metrology},\ }}\href {https://doi.org/10.1103/PhysRevA.94.052108} {\bibfield
  {journal} {\bibinfo  {journal} {Phys. Rev. A}\ }\textbf {\bibinfo {volume}
  {94}},\ \bibinfo {pages} {052108} (\bibinfo {year} {2016})}\BibitemShut
  {NoStop}%
\bibitem [{\citenamefont {Yuan}(2016)}]{yuan2016sequential}%
  \BibitemOpen
  \bibfield  {author} {\bibinfo {author} {\bibfnamefont {H.}~\bibnamefont
  {Yuan}},\ }\bibfield  {title} {\emph {\bibinfo {title} {Sequential feedback
  scheme outperforms the parallel scheme for hamiltonian parameter
  estimation},\ }}\href {https://doi.org/10.1103/PhysRevLett.117.160801}
  {\bibfield  {journal} {\bibinfo  {journal} {Phys. Rev. Lett.}\ }\textbf
  {\bibinfo {volume} {117}},\ \bibinfo {pages} {160801} (\bibinfo {year}
  {2016})}\BibitemShut {NoStop}%
\bibitem [{\citenamefont {Kura}\ and\ \citenamefont {Ueda}(2018)}]{Kura2018}%
  \BibitemOpen
  \bibfield  {author} {\bibinfo {author} {\bibfnamefont {N.}~\bibnamefont
  {Kura}}\ and\ \bibinfo {author} {\bibfnamefont {M.}~\bibnamefont {Ueda}},\
  }\bibfield  {title} {\emph {\bibinfo {title} {Finite-error metrological
  bounds on multiparameter hamiltonian estimation},\ }}\href
  {https://doi.org/10.1103/PhysRevA.97.012101} {\bibfield  {journal} {\bibinfo
  {journal} {Phys. Rev. A}\ }\textbf {\bibinfo {volume} {97}},\ \bibinfo
  {pages} {012101} (\bibinfo {year} {2018})}\BibitemShut {NoStop}%
\bibitem [{\citenamefont {Liu}\ and\ \citenamefont {Yuan}(2017)}]{Liu2017}%
  \BibitemOpen
  \bibfield  {author} {\bibinfo {author} {\bibfnamefont {J.}~\bibnamefont
  {Liu}}\ and\ \bibinfo {author} {\bibfnamefont {H.}~\bibnamefont {Yuan}},\
  }\bibfield  {title} {\emph {\bibinfo {title} {Control-enhanced multiparameter
  quantum estimation},\ }}\href {https://doi.org/10.1103/PhysRevA.96.042114}
  {\bibfield  {journal} {\bibinfo  {journal} {Phys. Rev. A}\ }\textbf {\bibinfo
  {volume} {96}},\ \bibinfo {pages} {042114} (\bibinfo {year}
  {2017})}\BibitemShut {NoStop}%
\bibitem [{\citenamefont {Nichols}\ \emph {et~al.}(2018)\citenamefont
  {Nichols}, \citenamefont {Liuzzo-Scorpo}, \citenamefont {Knott},\ and\
  \citenamefont {Adesso}}]{Nichols2018}%
  \BibitemOpen
  \bibfield  {author} {\bibinfo {author} {\bibfnamefont {R.}~\bibnamefont
  {Nichols}}, \bibinfo {author} {\bibfnamefont {P.}~\bibnamefont
  {Liuzzo-Scorpo}}, \bibinfo {author} {\bibfnamefont {P.~A.}\ \bibnamefont
  {Knott}}, \ and\ \bibinfo {author} {\bibfnamefont {G.}~\bibnamefont
  {Adesso}},\ }\bibfield  {title} {\emph {\bibinfo {title} {Multiparameter
  gaussian quantum metrology},\ }}\href
  {https://doi.org/10.1103/PhysRevA.98.012114} {\bibfield  {journal} {\bibinfo
  {journal} {Phys. Rev. A}\ }\textbf {\bibinfo {volume} {98}},\ \bibinfo
  {pages} {012114} (\bibinfo {year} {2018})}\BibitemShut {NoStop}%
\bibitem [{\citenamefont {Ge}\ \emph {et~al.}(2018)\citenamefont {Ge},
  \citenamefont {Jacobs}, \citenamefont {Eldredge}, \citenamefont {Gorshkov},\
  and\ \citenamefont {Foss-Feig}}]{ge2018distributed}%
  \BibitemOpen
  \bibfield  {author} {\bibinfo {author} {\bibfnamefont {W.}~\bibnamefont
  {Ge}}, \bibinfo {author} {\bibfnamefont {K.}~\bibnamefont {Jacobs}}, \bibinfo
  {author} {\bibfnamefont {Z.}~\bibnamefont {Eldredge}}, \bibinfo {author}
  {\bibfnamefont {A.~V.}\ \bibnamefont {Gorshkov}}, \ and\ \bibinfo {author}
  {\bibfnamefont {M.}~\bibnamefont {Foss-Feig}},\ }\bibfield  {title} {\emph
  {\bibinfo {title} {Distributed quantum metrology with linear networks and
  separable inputs},\ }}\href {https://doi.org/10.1103/PhysRevLett.121.043604}
  {\bibfield  {journal} {\bibinfo  {journal} {Phys. Rev. Lett.}\ }\textbf
  {\bibinfo {volume} {121}},\ \bibinfo {pages} {043604} (\bibinfo {year}
  {2018})}\BibitemShut {NoStop}%
\bibitem [{\citenamefont {Braunstein}\ and\ \citenamefont
  {Caves}(1994)}]{braunstein1994statistical}%
  \BibitemOpen
  \bibfield  {author} {\bibinfo {author} {\bibfnamefont {S.~L.}\ \bibnamefont
  {Braunstein}}\ and\ \bibinfo {author} {\bibfnamefont {C.~M.}\ \bibnamefont
  {Caves}},\ }\bibfield  {title} {\emph {\bibinfo {title} {Statistical distance
  and the geometry of quantum states},\ }}\href
  {https://doi.org/10.1103/PhysRevLett.72.3439} {\bibfield  {journal} {\bibinfo
   {journal} {Phys. Rev. Lett.}\ }\textbf {\bibinfo {volume} {72}},\ \bibinfo
  {pages} {3439} (\bibinfo {year} {1994})}\BibitemShut {NoStop}%
\bibitem [{\citenamefont {Helstrom}(1976)}]{helstrom1976quantum}%
  \BibitemOpen
  \bibfield  {author} {\bibinfo {author} {\bibfnamefont {C.~W.}\ \bibnamefont
  {Helstrom}},\ }\href@noop {} {\emph {\bibinfo {title} {Quantum detection and
  estimation theory}}}\ (\bibinfo  {publisher} {Academic press},\ \bibinfo
  {year} {1976})\BibitemShut {NoStop}%
\bibitem [{\citenamefont {Holevo}(1982)}]{Holevo1982}%
  \BibitemOpen
  \bibfield  {author} {\bibinfo {author} {\bibfnamefont {A.~S.}\ \bibnamefont
  {Holevo}},\ }\href@noop {} {\emph {\bibinfo {title} {Probabilistic and
  Statistical Aspects of Quantum Theory}}}\ (\bibinfo  {publisher} {North
  Holland, Amsterdam},\ \bibinfo {year} {1982})\BibitemShut {NoStop}%
\bibitem [{\citenamefont {Demkowicz-Dobrzanski}\ \emph
  {et~al.}(2020)\citenamefont {Demkowicz-Dobrzanski}, \citenamefont {Gorecki},\
  and\ \citenamefont {Guta}}]{demkowicz2020multi}%
  \BibitemOpen
  \bibfield  {author} {\bibinfo {author} {\bibfnamefont {R.}~\bibnamefont
  {Demkowicz-Dobrzanski}}, \bibinfo {author} {\bibfnamefont {W.}~\bibnamefont
  {Gorecki}}, \ and\ \bibinfo {author} {\bibfnamefont {M.}~\bibnamefont
  {Guta}},\ }\bibfield  {title} {\emph {\bibinfo {title} {Multi-parameter
  estimation beyond quantum fisher information},\ }}\href
  {http://doi.org/10.1088/1751-8121/ab8ef3} {\bibfield  {journal} {\bibinfo
  {journal} {Journal of Physics A: Mathematical and Theoretical}\ } (\bibinfo
  {year} {2020})}\BibitemShut {NoStop}%
\bibitem [{\citenamefont {Nagaoka}\ and\ \citenamefont
  {Hayashi}(2005)}]{nagaoka2005asymptotic}%
  \BibitemOpen
  \bibfield  {author} {\bibinfo {author} {\bibfnamefont {H.}~\bibnamefont
  {Nagaoka}}\ and\ \bibinfo {author} {\bibfnamefont {M.}~\bibnamefont
  {Hayashi}},\ }\href@noop {} {\emph {\bibinfo {title} {Asymptotic Theory of
  Quantum Statistical Inference}}}\ (\bibinfo  {publisher} {World Scientific
  Singapore},\ \bibinfo {year} {2005})\ Chap.~\bibinfo {chapter}
  {8}\BibitemShut {NoStop}%
\bibitem [{\citenamefont {Suzuki}(2016)}]{suzuki2016explicit}%
  \BibitemOpen
  \bibfield  {author} {\bibinfo {author} {\bibfnamefont {J.}~\bibnamefont
  {Suzuki}},\ }\bibfield  {title} {\emph {\bibinfo {title} {Explicit formula
  for the holevo bound for two-parameter qubit-state estimation problem},\
  }}\href {https://doi.org/10.1063/1.4945086} {\bibfield  {journal} {\bibinfo
  {journal} {J. Math. Phys.}\ }\textbf {\bibinfo {volume} {57}},\ \bibinfo
  {pages} {042201} (\bibinfo {year} {2016})}\BibitemShut {NoStop}%
\bibitem [{\citenamefont {Gu{\c{t}}{\u{a}}}\ and\ \citenamefont
  {Jen{\v{c}}ov{\'a}}(2007)}]{Guta2007}%
  \BibitemOpen
  \bibfield  {author} {\bibinfo {author} {\bibfnamefont {M.}~\bibnamefont
  {Gu{\c{t}}{\u{a}}}}\ and\ \bibinfo {author} {\bibfnamefont {A.}~\bibnamefont
  {Jen{\v{c}}ov{\'a}}},\ }\bibfield  {title} {\emph {\bibinfo {title} {Local
  asymptotic normality in quantum statistics},\ }}\href
  {https://doi.org/10.1007/s00220-007-0340-1} {\bibfield  {journal} {\bibinfo
  {journal} {Comm. Math. Phys.}\ }\textbf {\bibinfo {volume} {276}},\ \bibinfo
  {pages} {341} (\bibinfo {year} {2007})}\BibitemShut {NoStop}%
\bibitem [{\citenamefont {Yamagata}\ \emph {et~al.}(2013)\citenamefont
  {Yamagata}, \citenamefont {Fujiwara}, \citenamefont {Gill} \emph
  {et~al.}}]{yamagata2013quantum}%
  \BibitemOpen
  \bibfield  {author} {\bibinfo {author} {\bibfnamefont {K.}~\bibnamefont
  {Yamagata}}, \bibinfo {author} {\bibfnamefont {A.}~\bibnamefont {Fujiwara}},
  \bibinfo {author} {\bibfnamefont {R.~D.}\ \bibnamefont {Gill}},  \emph
  {et~al.},\ }\bibfield  {title} {\emph {\bibinfo {title} {Quantum local
  asymptotic normality based on a new quantum likelihood ratio},\ }}\href
  {http://doi.org/10.1214/13-AOS1147} {\bibfield  {journal} {\bibinfo
  {journal} {Ann. Statist.}\ }\textbf {\bibinfo {volume} {41}},\ \bibinfo
  {pages} {2197} (\bibinfo {year} {2013})}\BibitemShut {NoStop}%
\bibitem [{\citenamefont {Fujiwara}(1994)}]{fujiwara1994multi}%
  \BibitemOpen
  \bibfield  {author} {\bibinfo {author} {\bibfnamefont {A.}~\bibnamefont
  {Fujiwara}},\ }\bibfield  {title} {\emph {\bibinfo {title} {Multi-parameter
  pure state estimation based on the right logarithmic derivative},\
  }}\href@noop {} {\bibfield  {journal} {\bibinfo  {journal} {Math. Eng. Tech.
  Rep}\ }\textbf {\bibinfo {volume} {94}},\ \bibinfo {pages} {94} (\bibinfo
  {year} {1994})}\BibitemShut {NoStop}%
\bibitem [{\citenamefont {Albarelli}\ \emph {et~al.}(2019)\citenamefont
  {Albarelli}, \citenamefont {Friel},\ and\ \citenamefont
  {Datta}}]{albarelli2019evaluating}%
  \BibitemOpen
  \bibfield  {author} {\bibinfo {author} {\bibfnamefont {F.}~\bibnamefont
  {Albarelli}}, \bibinfo {author} {\bibfnamefont {J.~F.}\ \bibnamefont
  {Friel}}, \ and\ \bibinfo {author} {\bibfnamefont {A.}~\bibnamefont
  {Datta}},\ }\bibfield  {title} {\emph {\bibinfo {title} {Evaluating the
  holevo cram\'er-rao bound for multiparameter quantum metrology},\ }}\href
  {https://doi.org/10.1103/PhysRevLett.123.200503} {\bibfield  {journal}
  {\bibinfo  {journal} {Phys. Rev. Lett.}\ }\textbf {\bibinfo {volume} {123}},\
  \bibinfo {pages} {200503} (\bibinfo {year} {2019})}\BibitemShut {NoStop}%
\bibitem [{\citenamefont {Lindblad}(1976)}]{lindblad1976generators}%
  \BibitemOpen
  \bibfield  {author} {\bibinfo {author} {\bibfnamefont {G.}~\bibnamefont
  {Lindblad}},\ }\bibfield  {title} {\emph {\bibinfo {title} {On the generators
  of quantum dynamical semigroups},\ }}\href
  {https://doi.org/10.1007/BF01608499} {\bibfield  {journal} {\bibinfo
  {journal} {Comm. Math. Phys.}\ }\textbf {\bibinfo {volume} {48}},\ \bibinfo
  {pages} {119} (\bibinfo {year} {1976})}\BibitemShut {NoStop}%
\bibitem [{\citenamefont {Gorini}\ \emph {et~al.}(1976)\citenamefont {Gorini},
  \citenamefont {Kossakowski},\ and\ \citenamefont
  {Sudarshan}}]{gorini1976completely}%
  \BibitemOpen
  \bibfield  {author} {\bibinfo {author} {\bibfnamefont {V.}~\bibnamefont
  {Gorini}}, \bibinfo {author} {\bibfnamefont {A.}~\bibnamefont {Kossakowski}},
  \ and\ \bibinfo {author} {\bibfnamefont {E.~C.~G.}\ \bibnamefont
  {Sudarshan}},\ }\bibfield  {title} {\emph {\bibinfo {title} {Completely
  positive dynamical semigroups of n-level systems},\ }}\href
  {https://doi.org/10.1063/1.522979} {\bibfield  {journal} {\bibinfo  {journal}
  {J. Math. Phys.}\ }\textbf {\bibinfo {volume} {17}},\ \bibinfo {pages} {821}
  (\bibinfo {year} {1976})}\BibitemShut {NoStop}%
\bibitem [{\citenamefont {Breuer}\ \emph {et~al.}(2002)\citenamefont {Breuer},
  \citenamefont {Petruccione} \emph {et~al.}}]{breuer2002theory}%
  \BibitemOpen
  \bibfield  {author} {\bibinfo {author} {\bibfnamefont {H.-P.}\ \bibnamefont
  {Breuer}}, \bibinfo {author} {\bibfnamefont {F.}~\bibnamefont {Petruccione}},
   \emph {et~al.},\ }\href@noop {} {\emph {\bibinfo {title} {The theory of open
  quantum systems}}}\ (\bibinfo  {publisher} {Oxford University Press on
  Demand},\ \bibinfo {year} {2002})\BibitemShut {NoStop}%
\bibitem [{\citenamefont {Kay}(1993)}]{Kay1993}%
  \BibitemOpen
  \bibfield  {author} {\bibinfo {author} {\bibfnamefont {S.~M.}\ \bibnamefont
  {Kay}},\ }\href@noop {} {\emph {\bibinfo {title} {Fundamentals of statistical
  signal processing: estimation theory}}}\ (\bibinfo  {publisher} {Prentice
  Hall},\ \bibinfo {year} {1993})\BibitemShut {NoStop}%
\bibitem [{\citenamefont {Gill}\ and\ \citenamefont
  {Massar}(2000)}]{gill2000state}%
  \BibitemOpen
  \bibfield  {author} {\bibinfo {author} {\bibfnamefont {R.}~\bibnamefont
  {Gill}}\ and\ \bibinfo {author} {\bibfnamefont {S.}~\bibnamefont {Massar}},\
  }\bibfield  {title} {\emph {\bibinfo {title} {State estimation for large
  ensembles},\ }}\href {https://doi.org/10.1103/PhysRevA.61.042312} {\bibfield
  {journal} {\bibinfo  {journal} {Phys. Rev. A}\ }\textbf {\bibinfo {volume}
  {61}},\ \bibinfo {pages} {042312} (\bibinfo {year} {2000})}\BibitemShut
  {NoStop}%
\bibitem [{\citenamefont {Knill}\ and\ \citenamefont
  {Laflamme}(1997)}]{knill1997theory}%
  \BibitemOpen
  \bibfield  {author} {\bibinfo {author} {\bibfnamefont {E.}~\bibnamefont
  {Knill}}\ and\ \bibinfo {author} {\bibfnamefont {R.}~\bibnamefont
  {Laflamme}},\ }\bibfield  {title} {\emph {\bibinfo {title} {Theory of quantum
  error-correcting codes},\ }}\href {https://doi.org/10.1103/PhysRevA.55.900}
  {\bibfield  {journal} {\bibinfo  {journal} {Phys. Rev. A}\ }\textbf {\bibinfo
  {volume} {55}},\ \bibinfo {pages} {900} (\bibinfo {year} {1997})}\BibitemShut
  {NoStop}%
\bibitem [{\citenamefont {Grant}\ and\ \citenamefont {Boyd}()}]{grant2008cvx}%
  \BibitemOpen
  \bibfield  {author} {\bibinfo {author} {\bibfnamefont {M.}~\bibnamefont
  {Grant}}\ and\ \bibinfo {author} {\bibfnamefont {S.}~\bibnamefont {Boyd}},\
  }\href {http://cvxr.com/cvx/} {\bibinfo {title} {Cvx: Matlab software for
  disciplined convex programming},\ }\BibitemShut {NoStop}%
\bibitem [{\citenamefont {Zhou}\ and\ \citenamefont
  {Jiang}(2020{\natexlab{b}})}]{zhou2020optimal}%
  \BibitemOpen
  \bibfield  {author} {\bibinfo {author} {\bibfnamefont {S.}~\bibnamefont
  {Zhou}}\ and\ \bibinfo {author} {\bibfnamefont {L.}~\bibnamefont {Jiang}},\
  }\bibfield  {title} {\emph {\bibinfo {title} {Optimal approximate quantum
  error correction for quantum metrology},\ }}\href {\doibase
  10.1103/PhysRevResearch.2.013235} {\bibfield  {journal} {\bibinfo  {journal}
  {Phys. Rev. Research}\ }\textbf {\bibinfo {volume} {2}},\ \bibinfo {pages}
  {013235} (\bibinfo {year} {2020}{\natexlab{b}})}\BibitemShut {NoStop}%
\bibitem [{\citenamefont {G\'orecki}\ \emph {et~al.}(2020)\citenamefont
  {G\'orecki}, \citenamefont {Demkowicz-Dobrza\ifmmode~\acute{n}\else
  \'{n}\fi{}ski}, \citenamefont {Wiseman},\ and\ \citenamefont
  {Berry}}]{Gorecki2020}%
  \BibitemOpen
  \bibfield  {author} {\bibinfo {author} {\bibfnamefont {W.}~\bibnamefont
  {G\'orecki}}, \bibinfo {author} {\bibfnamefont {R.}~\bibnamefont
  {Demkowicz-Dobrza\ifmmode~\acute{n}\else \'{n}\fi{}ski}}, \bibinfo {author}
  {\bibfnamefont {H.~M.}\ \bibnamefont {Wiseman}}, \ and\ \bibinfo {author}
  {\bibfnamefont {D.~W.}\ \bibnamefont {Berry}},\ }\bibfield  {title} {\emph
  {\bibinfo {title} {$\ensuremath{\pi}$-corrected heisenberg limit},\ }}\href
  {\doibase 10.1103/PhysRevLett.124.030501} {\bibfield  {journal} {\bibinfo
  {journal} {Phys. Rev. Lett.}\ }\textbf {\bibinfo {volume} {124}},\ \bibinfo
  {pages} {030501} (\bibinfo {year} {2020})}\BibitemShut {NoStop}%
\bibitem [{\citenamefont {Lidar}\ \emph {et~al.}(1998)\citenamefont {Lidar},
  \citenamefont {Chuang},\ and\ \citenamefont {Whaley}}]{lidar1998decoherence}%
  \BibitemOpen
  \bibfield  {author} {\bibinfo {author} {\bibfnamefont {D.~A.}\ \bibnamefont
  {Lidar}}, \bibinfo {author} {\bibfnamefont {I.~L.}\ \bibnamefont {Chuang}}, \
  and\ \bibinfo {author} {\bibfnamefont {K.~B.}\ \bibnamefont {Whaley}},\
  }\bibfield  {title} {\emph {\bibinfo {title} {Decoherence-free subspaces for
  quantum computation},\ }}\href {https://doi.org/10.1103/PhysRevLett.81.2594}
  {\bibfield  {journal} {\bibinfo  {journal} {Phys. Rev. Lett.}\ }\textbf
  {\bibinfo {volume} {81}},\ \bibinfo {pages} {2594} (\bibinfo {year}
  {1998})}\BibitemShut {NoStop}%
\end{thebibliography}

\end{document}